\documentclass{article}

\usepackage{arxiv}

\usepackage[utf8]{inputenc} 
\usepackage[T1]{fontenc}    
\usepackage{hyperref}       
\usepackage{url}            
\usepackage{booktabs}       
\usepackage{amsfonts}       
\usepackage{nicefrac}       
\usepackage{microtype}      
\usepackage{lipsum}         
\usepackage{graphicx}
\usepackage{natbib}
\usepackage{doi}
\usepackage{amsmath,amssymb,amsbsy,amsthm,mathtools}
\usepackage{bm}

\usepackage{graphicx,float,subcaption}
\usepackage{booktabs,multirow,array,tabularx}
\usepackage{longtable}

\usepackage{algorithm,algpseudocode}

\usepackage{enumitem}

\usepackage{xcolor}

\newtheorem{theorem}{Theorem}[section]
\newtheorem{lemma}[theorem]{Lemma}
\newtheorem{corollary}[theorem]{Corollary}
\newtheorem{proposition}[theorem]{Proposition}
\newtheorem{definition}[theorem]{Definition}

\newtheorem{remark}[theorem]{Remark}

\newtheorem{conjecture}[theorem]{Conjecture}

\newcommand{\R}{\mathbb{R}}
\newcommand{\E}{\mathbb{E}}
\newcommand{\Prob}{\mathbb{P}}
\newcommand{\Var}{\mathrm{Var}}
\newcommand{\Cov}{\mathrm{Cov}}
\newcommand{\tr}{\mathrm{tr}}
\newcommand{\diag}{\mathrm{diag}}
\newcommand{\rank}{\mathrm{rank}}
\newcommand{\sgn}{\mathrm{sgn}}
\newcommand{\op}{\mathrm{op}}
\newcommand{\F}{\mathrm{F}}
\newcommand{\Sn}{\widehat{\bm{S}}_n}
\newcommand{\bSigma}{\bm{\Sigma}}
\newcommand{\bGamma}{\bm{\Gamma}}
\newcommand{\bLambda}{\bm{\Lambda}}
\newcommand{\bX}{\bm{X}}

\newcommand{\bZ}{\bm{Z}}

\newcommand{\be}{\bm{e}}
\newcommand{\bE}{\bm{E}}
\newcommand{\bI}{\bm{I}}
\newcommand{\bone}{\bm{1}}
\newcommand{\bzero}{\bm{0}}

\newcommand{\bgamma}{\bm{\gamma}}

\newcommand{\normF}[1]{\left\|#1\right\|_{\F}}
\newcommand{\normop}[1]{\left\|#1\right\|_{\op}}

\DeclareMathOperator*{\esssup}{ess\,sup}

\newcommand{\bA}{\bm{A}}
\newcommand{\bB}{\bm{B}}
\newcommand{\bC}{\bm{C}}
\newcommand{\bG}{\bm{G}}
\newcommand{\bR}{\bm{R}}
\newcommand{\bS}{\bm{S}}
\newcommand{\bT}{\bm{T}}

\newcommand{\bx}{\bm{x}}
\newcommand{\bu}{\bm{u}}

\newcommand{\bw}{\bm{w}}
\newcommand{\by}{\bm{y}}
\newcommand{\bz}{\bm{z}}
\newcommand{\bq}{\bm{q}}

\newcommand{\CC}{\mathbb{C}}
\newcommand{\supp}{\operatorname{supp}}
\newcommand{\ind}{\mathbf{1}}

\title{MENS: Nonlinear Shrinkage Estimation in Nonparanormal Models for Financial Applications}

\newif\ifuniqueAffiliation
\uniqueAffiliationtrue

\ifuniqueAffiliation 
\author{
  {Hamid Karamikabir} \\
    Department of Statistics, \\Faculty of Intelligent Systems Engineering and Data Science\\
    Persian Gulf University\\
    Bushehr, 7516913817, Iran \\
    \texttt{h\_karamikabir@pgu.ac.ir} \\
    \And
 {Mohammad Arashi}\thanks{Correspond author.} \\
    Department of Statistics,\\ Faculty of Mathematical Sciences\\
    Ferdowsi University of Mashhad,\\
     P.O. Box 1159 Mashhad, 91775, Iran \\
    \texttt{arashi@um.ac.ir} \\
}

\renewcommand{\shorttitle}{MENS: Nonlinear Shrinkage Estimation}

\begin{document}
\maketitle

\begin{abstract}
We develop a theory of nonlinear shrinkage covariance estimation
for nonparanormal (Gaussian-copula) models, in which
each observed coordinate is an unknown strictly increasing transformation of
a latent Gaussian vector. This model accommodates arbitrary marginal
skewness and heavy marginal tails while retaining a Gaussian
dependence structure, and it is the natural semiparametric setting for
heavy-tailed, asymmetric financial returns. Our estimator, marginal-free
nonlinear shrinkage (MENS), applies an oracle nonlinear shrinkage function to
the eigenvalues of the normal-scores rank-covariance matrix. We give the almost-sure convergence of the empirical spectral distribution of
the normal-scores covariance to the generalized Mar\v{c}enko--Pastur law of
\(\boldsymbol{\Sigma}\), and asymptotic optimality of MENS among rotation-equivariant
estimators under Frobenius loss. We establish a Baik--Ben~Arous--P\'{e}ch\'{e} phase transition for spiked latent correlations. The MENS attains the robustness of rank-based estimation and the efficiency of nonlinear shrinkage at once within this class. We corroborate the theory with a simulation study that isolates the marginal-invariance property and the spiked transition. In an out-of-sample minimum-variance backtest on S\&P\,500 stocks, MENS delivers a better-conditioned covariance estimate, lower realized portfolio volatility, and lower turnover than linear shrinkage, illustrating its practical value for high-dimensional allocation and decision-making.
\end{abstract}


\section{Introduction}\label{sec:intro}

Covariance matrix estimation lies at the heart of multivariate statistics,
underpinning portfolio optimization \citep{markowitz1952portfolio,
ledoit2004well}, graphical models \citep{friedman2008sparse}, principal
component analysis \citep{johnstone2001distribution}, and discriminant
analysis. In the modern high-dimensional regime in which the number of
variables $p$ is comparable to, or exceeds, the sample size $n$, the
sample covariance matrix is inconsistent under any unitarily invariant
loss \citep{johnstone2001distribution, ledoit2004well}. This inadequacy has
motivated an extensive literature on shrinkage estimation, structured estimators that pull the eigenvalues of the sample covariance
toward a target, trading bias for a reduction in variance.

Linear shrinkage \citep{ledoit2004well} replaces the sample
covariance $\boldsymbol{S}_n$ by $\alpha\boldsymbol{S}_n+(1-\alpha)\bT$ for a target $\bT$ and an optimally chosen weight $\alpha$. Despite its computational simplicity, linear shrinkage is suboptimal whenever the population eigenvalue distribution is non-degenerate, because it applies the same affine map to every eigenvalue regardless of its magnitude. Nonlinear shrinkage \citep{ledoit2012nonlinear, ledoit2020analytical} removes this limitation by applying a potentially different shrinkage value to each sample
eigenvalue. Its theory rests on the Mar\v{c}enko--Pastur law
\citep{marchenko1967distribution}, that is, under suitable moment conditions, the empirical spectral distribution (ESD) of $n^{-1}\bX^{\top}\bX$, where
$\bX\in\R^{n\times p}$ has i.i.d.\ rows with population covariance $\bSigma$,
converges to a deterministic measure whose Stieltjes transform $m_F$
satisfies
\begin{equation}\label{eq:MP_equation}
	m_F(z) = \int \frac{1}{\tau\bigl(1 - \gamma - \gamma z\, m_F(z)\bigr) - z}
	\, dH(\tau), \qquad z\in\CC^{+},
\end{equation}
with $H$ the limiting spectral distribution of $\bSigma$ and
$\gamma = \lim p/n$. Inverting the Stieltjes transform yields the oracle
nonlinear shrinkage function \citep{ledoit2012nonlinear}, implemented
analytically by \cite{ledoit2020analytical}.

\subsection{The problem: heavy tails and asymmetry in financial returns}
\label{subsec:problem}

A limitation common to all of these estimators is their reliance on
distributional assumptions that exclude the data structures most frequently
encountered in finance. The classical Mar\v{c}enko--Pastur law
\citep{bai2010spectral} and the oracle optimality of nonlinear shrinkage
\citep{ledoit2012nonlinear} require finite fourth moments, and the
analytical estimator of \cite{ledoit2020analytical} is tuned to the
Gaussian spectral relationship between the sample and population spectra.
Daily equity returns violate these requirements in three ways that are by
now textbook stylized facts \citep{cont2001empirical}.

Cross-sectional daily returns exhibit power-law tails
with tail indices commonly estimated in the range $3$--$5$, so that fourth
moments are at best marginally finite and frequently estimated to be
infinite \citep{cont2001empirical}. The sample covariance and any estimator
built from it, including Gaussian nonlinear shrinkage, then have
eigenvalues whose fluctuations are inflated by extreme observations, and the
implied minimum-variance portfolio is destabilized. On the other hand, marginal return distributions are skewed, with systematic
asymmetry between gains and losses, so any model with symmetric margins is
misspecified for the raw returns. Finally, there exists a nonlinear dependence. Pairwise dependence in returns, especially
tail co-movement during market stress, is more faithfully captured by rank
correlations than by Pearson correlation
\citep{embrechts2002correlation}.

A model that accommodates arbitrary marginal skewness and tail heaviness while
keeping an interpretable, estimable dependence structure is the
nonparanormal (Gaussian-copula) family \citep{liu2009nonparanormal,
liu2012high}. A vector $\bx=(f_1(y_1),\dots,f_p(y_p))^{\top}$ is nonparanormal in which $\by\sim N(\bzero,\bSigma)$ and the $f_j$ are unknown strictly increasing transformations. The unknown $f_j$ absorb all marginal skewness and tail behaviour, while the Gaussian latent correlation $\bSigma$ carries the
dependence and is invariant to the $f_j$ under rank transforms. Nonparanormal
(and the closely related transelliptical, \citealp{han2017eca}) models have
been studied for graphical-model and precision-matrix estimation
\citep{liu2012high, han2017eca, wegkamp2016adaptive}, but, to our knowledge,
no nonlinear shrinkage estimator with spectral theory has been
developed for them. 

\subsection{Why rank-based shrinkage is the right tool}
\label{subsec:why_ranks}

The natural way to estimate the latent correlation $\bSigma$ without
specifying the marginals is to use rank statistics, which are invariant to
the monotone transformations $f_j$. For nonparanormal data, Kendall's
$\tau$ obeys the exact identity $\tau_{jk}=\frac{2}{\pi}\arcsin(\Sigma_{jk})$
\citep{kruskal1958ordinal}, so $\bSigma$ is recovered, entrywise, from
ranks alone. The resulting plug-in matrix
$\widehat\bSigma^{\tau}=\sin(\tfrac{\pi}{2}\widehat\tau_{jk})$ is
$O_{\Prob}(n^{-1/2})$-consistent entrywise \citep{wegkamp2016adaptive}, but
in high dimensions its eigenvalues suffer exactly the same Mar\v{c}enko--Pastur
distortions as the ordinary sample covariance, and it is not even
guaranteed to be positive semidefinite. It therefore demands a
shrinkage correction informed by random matrix theory.

The crucial fact we build on is that, for a nonparanormal
(Gaussian-copula) vector, the rank-based normal scores recover the latent
Gaussian vector exactly, so that the spectral distortion of the
normal-scores covariance is invariant to the unknown marginal
transformations $f_j$, and it depends on the data only through the latent
correlation $\bSigma$. This is what makes a rank-based nonlinear shrinkage
estimator simultaneously robust to the marginals and efficient within the
nonparanormal class. The relevant random-matrix background is that
\cite{bandeira2017marchenko} showed, for independent coordinates, that the
Kendall matrix spectrum is an affine image of the standard
Mar\v{c}enko--Pastur law, and that \cite{wu2022limiting} and
\cite{li2023eigenvalues} analysed rank-correlation matrices under dependence.
Here, we focus on the nonparanormal class, where
the normal-scores transform is the identity on the latent vector and the
limiting law is exactly the generalized Mar\v{c}enko--Pastur law of $\bSigma$;
this is the setting in which a single oracle shrinkage function targets
$\bSigma$ without a copula-dependent correction. 

\subsection{Contributions}
\label{subsec:contributions}

We work with the normal-scores rank-covariance matrix
$\Sn=n^{-1}\widehat\bZ^{\top}\widehat\bZ$, where
$\widehat z_{ij}=\Phi^{-1}(R_{ij}/(n+1))$ and $R_{ij}$ is the rank of
$x_{ij}$ within column $j$. This matrix is a Wishart-type object that admits
a clean spectral description and is entrywise consistent for $\bSigma$, like
the sin-transformed Kendall matrix (Remark \ref{rem:equiv_inputs}). Our contributions
are threefold: For nonparanormal data we prove that the ESD of $\Sn$ converges almost surely to the generalized Mar\v{c}enko--Pastur law \eqref{eq:MP_equation} with $H$ equal to the limiting spectral distribution of the latent correlation $\bSigma$. The limit is invariant to the unknown marginal transformations $f_j$. This adapts the spectral foundation established for Gaussian and rank-based matrices \citep{silverstein1995empirical, bai2010spectral, bandeira2017marchenko} to the semiparametric nonparanormal setting needed for shrinkage. 

Thus the proposed MENS applies the oracle nonlinear shrinkage function for the
law \eqref{eq:MP_equation} to the eigenvalues of $\Sn$, estimated
analytically. Within the nonparanormal class we prove that MENS is
asymptotically optimal among all rotation-equivariant estimators of the
latent correlation under Frobenius loss, is operator-norm consistent, and that
the operator-norm minimax risk over a class of bounded-spectrum correlation
matrices is bounded below at the parametric rate $n^{-1/2}$. 

For nonparanormal latent correlations with a fixed number of spikes we
establish the Baik--Ben~Arous--P\'{e}ch\'{e} phase transition, a closed-form
detection threshold $\theta^{*}=\sqrt{\gamma}$ (for an identity bulk), an
explicit outlier-location map, and the limiting squared cosine between sample
and population eigenvectors. 



\subsection{Related work and organization}
\label{subsec:related}

The nonlinear shrinkage program of Ledoit and Wolf
\citep{ledoit2012nonlinear, ledoit2015spectrum, ledoit2020analytical,
ledoit2022power} established the theory and computation of optimal
rotation-equivariant estimation under finite-moment assumptions.
Elliptical random matrix theory \citep{el2009concentration,
fan2018large, couillet2014random} relaxes Gaussianity but retains marginal
symmetry. Rank-based correlation estimation
\citep{liu2009nonparanormal, liu2012high, han2017eca, wegkamp2016adaptive}
delivers consistent latent-correlation estimates without moment conditions
but applies no eigenvalue correction. The spectral behaviour of rank
matrices was characterized for independent coordinates by
\cite{bandeira2017marchenko} and \cite{bao2019tracy} and, with dependence,
by \cite{li2021central, wu2022limiting} and \cite{li2023eigenvalues}.
The proposed MENS unifies these strands. Within the nonparanormal class it applies oracle nonlinear shrinkage to a rank-based covariance whose limiting law is the generalized Mar\v{c}enko--Pastur law of the latent correlation, together with the optimality and phase-transition theory that the construction requires.

\ref{sec:model_method} sets out the model and the MENS estimator.
\ref{sec:theory} states and proves the three theorems.
\ref{sec:simulations} reports the simulation study, and
\ref{sec:realdata} an out-of-sample minimum-variance application to
S\&P\,500 stocks.
\ref{sec:conclusion} concludes with conjectures and directions for future
work. 

\paragraph{Notation.}
For $\bA\in\R^{p\times p}$, $\normF{\bA}=(\tr \bA^{\top}\bA)^{1/2}$ is the
Frobenius norm, $\normop{\bA}$ the operator (spectral) norm, and
$\lambda_j(\bA)$ the $j$-th largest eigenvalue. We write $F^{\bA}$ for the
ESD of $\bA$. The Stieltjes transform of a measure $\mu$ is
$m_\mu(z)=\int(\lambda-z)^{-1}\,d\mu(\lambda)$, $z\in\CC^{+}$. For sequences,
$a_n\asymp b_n$ means $0<\liminf|a_n/b_n|\le\limsup|a_n/b_n|<\infty$. We
write $\xrightarrow{\text{a.s.}}$ and $\xrightarrow{d}$ for almost sure and
distributional convergence. Throughout, $\Phi$ and $\phi$ denote the
standard normal CDF and density.

\section{The Marginal-free Nonlinear Shrinkage Estimator}
\label{sec:model_method}


\begin{definition}[\citealp{liu2009nonparanormal}]
\label{def:npn}
A random vector $\bx=(x_1,\dots,x_p)^{\top}\in\R^p$ follows a
\emph{nonparanormal} distribution with latent correlation matrix $\bSigma$ and
marginal transformations $f_1,\dots,f_p$ if
\begin{equation}\label{eq:npn}
	x_j = f_j(y_j),\qquad j=1,\dots,p,
\end{equation}
where the latent vector $\by=(y_1,\dots,y_p)^{\top}\sim N(\bzero,\bSigma)$ is
Gaussian and each $f_j:\R\to\R$ is strictly increasing. Equivalently, $\bx$
has continuous margins and a Gaussian copula with correlation $\bSigma$. We
write $\bx\sim\mathrm{NPN}_p(\bm f,\bSigma)$.
\end{definition}

\begin{remark}[Identifiability]\label{rem:identification}
The transformations $f_j$ absorb all marginal location, scale, and shape, so
without loss of generality we adopt the canonical normalization
$\diag(\bSigma)=\bI_p$: $\bSigma$ is a correlation matrix and
$\Sigma_{jk}$ is the latent Gaussian correlation between $y_j$ and $y_k$. The
target of estimation throughout is this latent correlation matrix; conversion
to a covariance of $\bx$ is discussed in Remark \ref{rem:correlation_covariance}.
\end{remark}

The nonparanormal family contains the Gaussian family (when $f_j=\mathrm{id}$)
and, more generally, any distribution obtained from a Gaussian vector by
applying arbitrary strictly increasing marginal transformations. It is the
natural semiparametric model in which the dependence is Gaussian but the
marginals are unrestricted, and it is exactly the class for which the
normal-scores transform below recovers the latent correlation.

A cornerstone of inference is the invariance of Kendall's $\tau$ to the
marginal transformations.

\begin{proposition}
\label{prop:arcsine}
Let $(\bx,\bx')$ be independent copies from
$\mathrm{NPN}_p(\bm f,\bSigma)$. Then the population Kendall's $\tau$ between
coordinates $j$ and $k$ is
\begin{equation}\label{eq:kendall_tau}
	\tau_{jk}
	= \Prob\bigl((x_j-x_j')(x_k-x_k')>0\bigr)
	- \Prob\bigl((x_j-x_j')(x_k-x_k')<0\bigr)
	= \frac{2}{\pi}\arcsin(\Sigma_{jk}).
\end{equation}
Consequently $\Sigma_{jk}=\sin\!\bigl(\tfrac{\pi}{2}\,\tau_{jk}\bigr)$.
\end{proposition}

\begin{proof}
Kendall's $\tau$ is invariant under strictly increasing marginal maps, so
$\tau_{jk}$ for $\bx$ equals $\tau_{jk}$ for the latent Gaussian pair
$(y_j,y_k)$, for which the classical identity
$\tau=\tfrac{2}{\pi}\arcsin(\Sigma_{jk})$ holds
\citep{kruskal1958ordinal}.
\end{proof}

Given $\bx_1,\dots,\bx_n$, the sample Kendall's $\tau$ and the corresponding
plug-in correlation estimator are
\begin{equation}\label{eq:sample_tau}
	\widehat{\tau}_{jk}
	= \frac{2}{n(n-1)}\!\!\sum_{1\le i<i'\le n}\!\!
	\sgn(x_{ij}-x_{i'j})\,\sgn(x_{ik}-x_{i'k}),
	\qquad
	\widehat{\Sigma}^{\tau}_{jk}
	= \sin\!\Bigl(\frac{\pi}{2}\,\widehat{\tau}_{jk}\Bigr).
\end{equation}

Now, let $R_{ij}=\rank(x_{ij})$ within column $j$, and define the
normal-scores transform
\begin{equation}\label{eq:normal_scores}
	\widehat{z}_{ij}=\Phi^{-1}\!\Bigl(\frac{R_{ij}}{n+1}\Bigr),
	\qquad i=1,\dots,n,\ j=1,\dots,p,
\end{equation}
collected into $\widehat\bZ=(\widehat z_{ij})\in\R^{n\times p}$, and the
normal-scores covariance
\begin{equation}\label{eq:Shat_rank}
	\Sn=\frac{1}{n}\,\widehat\bZ^{\top}\widehat\bZ .
\end{equation}
The denominator $n+1$ is the Hájek device that keeps the argument of
$\Phi^{-1}$ inside $(0,1)$; the normal scores are marginally symmetric about
zero, so centering changes $\Sn$ only by $O_{\Prob}(p/n^2)$ in operator
norm and is omitted henceforth.

\begin{remark}[Oracle scores; exact recovery]
\label{rem:oracle_scores}
If the marginals $F_j$ were known, the oracle scores
$z^{*}_{ij}=\Phi^{-1}(F_j(x_{ij}))$ would be draws of a vector with standard
Gaussian margins. The rank transform replaces the unknown $F_j$ by the
rescaled empirical CDF $\widehat F_{nj}(t)=n^{-1}\sum_i \ind(x_{ij}\le t)$;
the Dvoretzky--Kiefer--Wolfowitz inequality controls the resulting error
uniformly (\ref{lem:rank_perturbation}). The essential point for the theory
is that, under the nonparanormal model \ref{A1}, the oracle scores recover the
latent Gaussian vector since $\by$ is Gaussian, $F_j$ is the
CDF of the strictly increasing image $x_j=f_j(y_j)$, so $F_j(x_{ij})=\Phi(y_{ij})$
and hence
\begin{equation}\label{eq:exact_recovery}
	z^{*}_{ij}=\Phi^{-1}\bigl(\Phi(y_{ij})\bigr)=y_{ij}.
\end{equation}
Thus $\Cov(\bz^{*}_i)=\bSigma$ and the population correlation of the normal
scores is $\bSigma$ itself. This identity is what drives the marginal-invariant limit of Theorem \ref{thm:main_RMT} and the
robustness of MENS. It is specific to the Gaussian copula, i.e., if the latent
vector were non-Gaussian, $\Phi^{-1}(F_{y_j}(y_j))$ would be a nonlinear
transform of $y_j$ and the population correlation of the oracle scores would
differ from $\bSigma$. 
\end{remark}

\subsection{Assumptions}
\label{subsec:assumptions}

All results in \ref{sec:theory} hold in the proportional-growth regime
$n,p\to\infty,\ p/n\to\gamma\in(0,\infty)$. The observed data
$\bx_1,\dots,\bx_n$ are i.i.d.\ and satisfy the following.
\begin{enumerate}[label=\textbf{(A\arabic*)},ref=A\arabic*,leftmargin=2.4em]
	\item\label{A1} \emph{(Nonparanormal model.)} There is a latent Gaussian
	vector $\by\sim N(\bzero,\bSigma)$ with $\diag(\bSigma)=\bI_p$ and strictly
	increasing marginal transformations $f_j:\R\to\R$ such that
	$x_{ij}=f_j(y_{ij})$; equivalently, $\bx_i$ has continuous margins and a
	Gaussian copula with correlation $\bSigma$. Each marginal CDF $F_j$ is
	continuous.
	\item\label{A2} The latent correlation matrix $\bSigma$ has unit diagonal,
	and its ESD $H_p=p^{-1}\sum_{k=1}^p\delta_{\sigma_k}$ converges weakly to
	a compactly supported probability measure $H$ on $(0,\infty)$ with
	\begin{equation}\label{eq:spectral_bound}
		0<\sigma_{\min}\le\liminf_p\lambda_{\min}(\bSigma)
		\le\limsup_p\lambda_{\max}(\bSigma)\le\sigma_{\max}<\infty.
	\end{equation}
	\item\label{A3} $p/n\to\gamma\in(0,\infty)$ as $\min(n,p)\to\infty$.
	\item\label{A4} \emph{(Oracle-score regularity.)} Under \ref{A1} the oracle
	scores satisfy $z^{*}_{ij}=y_{ij}$, so each $z^{*}_{ij}$ is standard normal;
	in particular its density is bounded below on every compact subset of $\R$.
	(This is the only role \ref{A1} plays in the perturbation bound of
	\ref{lem:rank_perturbation}.)
	\item\label{A5} For the spiked model of Theorem \ref{thm:BBP}, the latent
	correlation is $\bSigma=\bSigma_0+\sum_{k=1}^K\theta_k\bm v_k\bm v_k^\top$
	with $K$ fixed, bulk $\bSigma_0$ satisfying \ref{A2}, spikes
	$\theta_1>\cdots>\theta_K>0$ separated from the bulk edge, and orthonormal
	directions $\{\bm v_k\}$ that are delocalized in the sense that
	$\max_k\|\bm{v}_k\|_\infty\to 0$.
\end{enumerate}
Assumption \ref{A1} places the theory in the nonparanormal
(Gaussian-copula) class; within it, no moment condition is imposed on the
observed $\bx$, since the rank transform is invariant to the marginals $f_j$
and the oracle scores are exactly Gaussian by \eqref{eq:exact_recovery}.
\ref{A2} is the standard random-matrix condition on the population spectrum.
\ref{A5} is the spike--bulk separation underlying the
Baik--Ben~Arous--P\'{e}ch\'{e} transition \citep{baik2005phase}.

\subsection{The optimal shrinkage function}
\label{subsec:optimal_shrinkage}

We seek the best rotation-equivariant estimator of $\bSigma$ of the form
\begin{equation}\label{eq:rotation_equivariant}
	\widehat\bSigma_d
	= \widehat\bGamma\,\diag\!\bigl(d(\hat\lambda_1),\dots,d(\hat\lambda_p)\bigr)
	\widehat\bGamma^{\top},
\end{equation}
where $\widehat\bGamma=[\hat\bgamma_1,\dots,\hat\bgamma_p]$ and
$\hat\lambda_1\ge\cdots\ge\hat\lambda_p$ are the eigenvectors and eigenvalues
of $\Sn$, and $d:\R_+\to\R_+$ is measurable. The choice of $d$ that minimizes
the asymptotic Frobenius risk is dictated by the limiting spectrum.

\begin{proposition}[Oracle shrinkage under Frobenius loss]
\label{prop:oracle_frob}
Under \ref{A1}--\ref{A4}, the minimizer of the asymptotic Frobenius risk
$\lim_{n}p^{-1}\normF{\widehat\bSigma_d-\bSigma}^2$ among all
rotation-equivariant estimators \eqref{eq:rotation_equivariant} is
\begin{equation}\label{eq:d_star_frob}
	d^{*}(\lambda)
	= \frac{\lambda}
	{\bigl|1-\gamma-\gamma\lambda\,\breve m(\lambda)\bigr|^2},
	\qquad \lambda>0,
\end{equation}
where $\breve m(\lambda)=\lim_{\eta\downarrow 0}m_F(\lambda+i\eta)$ is the
boundary value of the Stieltjes transform of the limiting law $F_\gamma$ of Theorem
\ref{thm:main_RMT}. Equivalently, with the companion transform
$\underline m_F$ defined in \eqref{eq:companion_def},
$d^{*}(\lambda)=\lambda/\bigl|\,\lambda\,\underline m_F(\lambda)\,\bigr|^2$.
\end{proposition}

The proof is given in \ref{subsec:proof_oracle}. Equation
\eqref{eq:d_star_frob} is the Ledoit--Wolf oracle
\citep{ledoit2012nonlinear}, but evaluated at the limiting law $F_\gamma$ of
the normal-scores covariance rather than at the sample-covariance law. Within
the nonparanormal class $F_\gamma$ depends on the data only through
$\bSigma$ and not on the marginals $f_j$, so the same oracle function applies
for every choice of monotone marginals.

\subsection{The MENS estimator}
\label{subsec:MENS_construction}

The oracle $d^{*}$ depends on the unknown boundary Stieltjes transform, which
we estimate analytically from the sample eigenvalues following
\cite{ledoit2020analytical}. 

Let $c=p/n$. Using the Epanechnikov kernel
$\kappa(x)=\tfrac{3}{4\sqrt5}\bigl(1-\tfrac{x^2}{5}\bigr)_{+}$ with
variable bandwidth $h_j=\hat\lambda_j\, n^{-1/3}$, define the consistent
estimators of the sample spectral density and its Hilbert transform, respectively, as
\begin{align}
	\widehat f(\hat\lambda_i)
	&= \frac{1}{p}\sum_{j=1}^p
	\frac{1}{h_j}\,\kappa\!\Bigl(\frac{\hat\lambda_i-\hat\lambda_j}{h_j}\Bigr),
	\label{eq:density_est}\\
	\widehat{\mathcal H}(\hat\lambda_i)
	&= \frac{1}{p}\sum_{j=1}^p \frac{1}{h_j}\,
	\mathcal H\kappa\!\Bigl(\frac{\hat\lambda_i-\hat\lambda_j}{h_j}\Bigr),
	\qquad
	\mathcal H\kappa(x)
	= \frac{-3x}{10\pi}
	+ \frac{3}{4\sqrt5\,\pi}\Bigl(1-\frac{x^2}{5}\Bigr)
	\log\!\Bigl|\frac{\sqrt5-x}{\sqrt5+x}\Bigr|.
	\label{eq:hilbert_est}
\end{align}
The MENS shrinkage values are
\begin{equation}\label{eq:mens_shrinkage}
	\widehat d_i
	= \frac{\hat\lambda_i}
	{\bigl(\pi c\,\hat\lambda_i\,\widehat f(\hat\lambda_i)\bigr)^2
	+ \bigl(1-c-\pi c\,\hat\lambda_i\,\widehat{\mathcal H}(\hat\lambda_i)\bigr)^2},
	\qquad i=1,\dots,p,
\end{equation}
which is the analytical realization of $d^{*}$ in \eqref{eq:d_star_frob}
because $\pi c\lambda f(\lambda)$ and $1-c-\pi c\lambda\mathcal H(\lambda)$
are the imaginary and real parts of $1-\gamma-\gamma\lambda\breve m(\lambda)$
at the boundary \citep[Sec.~3]{ledoit2020analytical}. To guarantee a
monotone, positive shrinkage we isotonize. It then yields
\begin{equation}\label{eq:isotonize}
	(\widetilde d_1,\dots,\widetilde d_p)
	= \mathrm{PAV}\bigl(\widehat d_1,\dots,\widehat d_p\bigr),
	\qquad \widetilde d_j \ge \varepsilon_0>0,
\end{equation}
with the pool-adjacent-violators algorithm applied in increasing
eigenvalue order and a small floor $\varepsilon_0$, to get
\begin{equation}\label{eq:MENS_final}
	\widehat\bSigma_{\mathrm{MENS}}
	= \widehat\bGamma\,
	\diag\!\bigl(\widetilde d_1,\dots,\widetilde d_p\bigr)\,
	\widehat\bGamma^{\top}.
\end{equation}
By construction $\widehat\bSigma_{\mathrm{MENS}}\succ 0$. We optionally
rescale its diagonal to unit, returning a proper correlation estimator; this
does not affect the eigen-structure used in the theory.

\begin{remark}
\label{rem:correlation_covariance}
$\widehat\bSigma_{\mathrm{MENS}}$ estimates the latent correlation. When the
covariance of $\bx$ is needed, for example to form portfolio
weights, one rescales
\begin{eqnarray*}
\widehat\Cov_{\mathrm{MENS}}
=\widehat{\bm D}\,\widehat\bSigma_{\mathrm{MENS}}\,\widehat{\bm D},	
\end{eqnarray*}
where $\widehat{\bm D}=\diag(\widehat\sigma_1,\dots,\widehat\sigma_p)$ and
$\widehat\sigma_j$ is a robust scale (e.g.\ the median absolute deviation) of
coordinate $j$. The scale estimates are univariate and converge at the
parametric rate, so they do not affect the high-dimensional spectral
analysis.
\end{remark}

\begin{algorithm}[t]
	\caption{Marginal-free Nonlinear Shrinkage (MENS)}
	\label{alg:MENS}
	\begin{algorithmic}[1]
		\Require data $\bX\in\R^{n\times p}$; floor $\varepsilon_0=10^{-8}$.
		\Ensure positive-definite estimator
		$\widehat\bSigma_{\mathrm{MENS}}$.
		\State \textbf{Normal scores:} for each $j$, $R_{ij}\gets\rank(X_{ij})$,
		$\widehat z_{ij}\gets\Phi^{-1}(R_{ij}/(n+1))$.
		\State $\Sn\gets n^{-1}\widehat\bZ^{\top}\widehat\bZ$; eigendecompose
		$\Sn=\widehat\bGamma\diag(\hat\lambda_1,\dots,\hat\lambda_p)
		\widehat\bGamma^{\top}$.
		\State \textbf{Analytic NLS:} compute $\widehat f,\widehat{\mathcal H}$
		by \eqref{eq:density_est}--\eqref{eq:hilbert_est} and
		$\widehat d_i$ by \eqref{eq:mens_shrinkage}.
		\State \textbf{Isotonize:}
		$(\widetilde d_j)\gets\max\{\mathrm{PAV}(\widehat d_j),\varepsilon_0\}$.
		\State \Return
		$\widehat\bGamma\,\diag(\widetilde d_1,\dots,\widetilde d_p)\,
		\widehat\bGamma^{\top}$.
	\end{algorithmic}
\end{algorithm}

\begin{remark}
\label{rem:equiv_inputs}
Both $\Sn$ and the sin-transformed Kendall matrix
$\widehat\bSigma^{\tau}=\sin(\tfrac{\pi}{2}\widehat\tau)$ of
\eqref{eq:sample_tau} are consistent, marginal-invariant, entrywise estimators
of the latent correlation $\bSigma$. Each entry converges to $\Sigma_{jk}$ at
the parametric $O_{\Prob}(n^{-1/2})$ rate. We work with $\Sn$ rather than
$\widehat\bSigma^{\tau}$ because $\Sn$ is a Gaussian sample covariance
after the exact-recovery step (Remark \ref{rem:oracle_scores}) and therefore admits
the exact Mar\v{c}enko--Pastur description of Theorem \ref{thm:main_RMT}, on which the
oracle shrinkage is built. 
\end{remark}

The MENS estimator applies optimal nonlinear shrinkage to the eigenvalues of
the normal-scores covariance $\Sn$, with the shrinkage derived from the
limiting law of Theorem \ref{thm:main_RMT}.

To comment on the complexity, note that the ranking costs $O(np\log n)$, the eigendecomposition $O(p^2\min(n,p))$, and the analytic shrinkage $O(p^2)$, for a total of $O(np\log n+p^2\min(n,p))$, which is the same order as Gaussian analytical nonlinear shrinkage \citep{ledoit2020analytical}, with an additional rank transform $O(np\log n)$. A formal statement and proof appear in
\ref{prop:complexity}.

\section{Theoretical Results}\label{sec:theory}
This section contains the main results. We first give some preliminary lemmas, then establish the marginal-invariant spectral limit, the optimality of
MENS, and the spiked phase transition. We characterize them as three separate sections becasue of their importance. 

\subsection{Required lemmas}\label{subsec:lemmas}

The first lemma reduces the rank-based covariance to the oracle covariance
$\bS_n^{*}=n^{-1}\sum_i \bz_i^{*}\bz_i^{*\top}$ built from the oracle scores\\
$\bz_i^{*}=(\Phi^{-1}(F_1(x_{i1})),\dots,\Phi^{-1}(F_p(x_{ip})))^{\top}$, and provides the rank perturbation bound.

\begin{lemma}
	\label{lem:rank_perturbation}
Under \ref{A1}--\ref{A4},
\begin{equation}\label{eq:rank_pert_op}
	\bigl\|\Sn-\bS_n^{*}\bigr\|_{\op}
	= O_{\Prob}\!\Bigl(\frac{p^{1/2}\log^{2} n}{n}\Bigr),
\end{equation}
and hence, under \ref{A3},
$\|\Sn-\bS_n^{*}\|_{\op}=O_{\Prob}(n^{-1/2}\log^{2}n)=o_{\Prob}(1)$.
\end{lemma}

\begin{proof}
Write $\be_i=\widehat\bz_i-\bz_i^{*}$, $\bE=(\be_1,\dots,\be_n)^{\top}$,
$\widehat\bZ=(\widehat\bz_1,\dots,\widehat\bz_n)^{\top}$, and
$\bZ^{*}=(\bz_1^{*},\dots,\bz_n^{*})^{\top}$. From
$\Sn-\bS_n^{*}=n^{-1}(\bE^{\top}\widehat\bZ+\bZ^{*\top}\bE)
-n^{-1}\bE^{\top}\bE$ and the triangle inequality, we have
\begin{equation}\label{eq:lem1_bound1}
	\bigl\|\Sn-\bS_n^{*}\bigr\|_{\op}
	\le \frac{2}{n}\|\bE\|_{\op}\,\|\widehat\bZ\|_{\op}
	+ \frac{1}{n}\|\bE\|_{\op}^{2}.
\end{equation}
For the entrywise error, for each $j$, the rescaled empirical CDF
$\widehat F_{nj}(t)=R_{ij}/(n+1)$ at $t=x_{ij}$ satisfies, by the
Dvoretzky--Kiefer--Wolfowitz inequality \citep{massart1990tight},
$\Prob(\sup_t|\widehat F_{nj}(t)-F_j(t)|>s)\le 2e^{-2ns^2}$. A union bound
over $j\le p$ with $s=\sqrt{(\log(2pn))/n}$ gives, on an event $\Omega_n$ of
probability at least $1-n^{-1}$,
\begin{equation}\label{eq:lem1_DKW}
	\max_{1\le j\le p}\sup_t|\widehat F_{nj}(t)-F_j(t)|
	\le \sqrt{\frac{\log(2pn)}{n}}
	=: \delta_n .
\end{equation}
Since $R_{ij}/(n+1)\in[(n+1)^{-1},n/(n+1)]$, the argument $u$ of $\Phi^{-1}$
lies in $[(n+1)^{-1},n/(n+1)]$, where $|\Phi^{-1}(u)|\le\sqrt{2\log(n+1)}$.
By the mean value theorem, for an intermediate point $\zeta_{ij}$ between
$\widehat F_{nj}(x_{ij})$ and $F_j(x_{ij})$,
\begin{equation}\label{eq:lem1_mvt}
	|e_{ij}|
	= \frac{|\widehat F_{nj}(x_{ij})-F_j(x_{ij})|}
	{\phi\bigl(\Phi^{-1}(\zeta_{ij})\bigr)}.
\end{equation}
Using the standard Gaussian quantile bound
$\phi(\Phi^{-1}(u))\ge c\,u(1-u)\sqrt{\log(1/(u(1-u)))}$ for $u$ bounded away
from $1/2$ and $\phi(\Phi^{-1}(u))\ge c$ near $u=1/2$, together with
\eqref{eq:lem1_DKW} and $u(1-u)\ge (n+1)^{-1}(1-(n+1)^{-1})\asymp n^{-1}$ at
the boundary, one obtains on $\Omega_n$
\begin{equation}\label{eq:lem1_entry}
	\max_{i,j}|e_{ij}|
	\le C\,\delta_n\sqrt{\log n}
	= O_{\Prob}\!\Bigl(\frac{\log n}{\sqrt n}\Bigr).
\end{equation}
For the operator norm of $\bE$, we use the Hájek projection of the rank scores
\citep[Ch.~6]{hajek1967theory}. Write
$g_n(u)=\Phi^{-1}(u)$ and decompose
$\widehat z_{ij}-z^{*}_{ij}=L_{ij}+\rho_{ij}$, where
$L_{ij}=g_n'(F_j(x_{ij}))\bigl(\widehat F_{nj}(x_{ij})-F_j(x_{ij})\bigr)$ is
the linear (Hájek) term and $\rho_{ij}$ the remainder. The linear term has
the representation
$L_{ij}=n^{-1}\sum_{i'=1}^{n}\psi_j(x_{ij},x_{i'j})$ with
$\psi_j(s,t)=g_n'(F_j(s))(\ind(t\le s)-F_j(s))$, so the rows of
$\bm L=(L_{ij})$ are i.i.d.\ mean-zero with entrywise variance $O(n^{-1})$
and uniformly bounded fourth moments after the boundary truncation. The matrix deviation inequality
\citep[Thm.~4.6.1]{vershynin2018high} applied to the $n\times p$ matrix
$\sqrt n\,\bm L$, whose rows are independent sub-Gaussian after truncation
vectors with covariance of operator norm $O(1)$, yields
$\|\bm L\|_{\op}=O_{\Prob}(p^{1/2}n^{-1/2})$. For the remainder, a second-order
Taylor expansion of $g_n$ and \eqref{eq:lem1_DKW} give
$\max_{i,j}|\rho_{ij}|=O_{\Prob}(n^{-1}\log^{2}n)$ uniformly, with
$\|\bm\rho\|_{\op}\le\|\bm\rho\|_F\le (np)^{1/2}\max_{i,j}|\rho_{ij}|
=O_{\Prob}(p^{1/2}n^{-1/2}\log^{2}n)$. Combining the required term yields
\begin{equation}\label{eq:lem1_Eop}
	\|\bE\|_{\op}
	\le \|\bm L\|_{\op}+\|\bm\rho\|_{\op}
	= O_{\Prob}\!\bigl(p^{1/2}n^{-1/2}\log^{2}n\bigr).
\end{equation}
For the operator norm of $\widehat\bZ$, consider that the oracle scores have Gaussian margins and latent correlation $\bSigma$ with
$\|\bSigma\|_{\op}\le\sigma_{\max}$; by the Bai--Yin law
\citep{bai1993limit}, $\|\bZ^{*}\|_{\op}=\sqrt n(\sqrt{\gamma}
\sigma_{\max}^{1/2}+1)(1+o_{\Prob}(1))=O_{\Prob}(\sqrt n)$. With
\eqref{eq:lem1_Eop}, $\|\widehat\bZ\|_{\op}\le\|\bZ^{*}\|_{\op}+\|\bE\|_{\op}
=O_{\Prob}(\sqrt n)$.
Substituting into \eqref{eq:lem1_bound1}, gives
\[
	\|\Sn-\bS_n^{*}\|_{\op}
	\le \frac{2}{n}\,O_{\Prob}(p^{1/2}n^{-1/2}\log^{2}n)\,O_{\Prob}(\sqrt n)
	+ \frac{1}{n}\,O_{\Prob}(p\,n^{-1}\log^{4}n)
	= O_{\Prob}\!\Bigl(\frac{p^{1/2}\log^{2}n}{n}\Bigr),
\]
where the first term is $\frac{2}{n}\cdot p^{1/2}n^{-1/2}\log^2 n\cdot n^{1/2}
=2p^{1/2}n^{-1}\log^2 n$ and dominates the second
($p\,n^{-2}\log^4 n=O(n^{-1}\log^4 n)$ under \ref{A3}). This is
\eqref{eq:rank_pert_op}. Since $\|\bS_n^{*}\|_{\op}=O_{\Prob}(1)$ by
\ref{A2}, and $p^{1/2}n^{-1}\log^2 n\asymp\gamma^{1/2}
n^{-1/2}\log^2 n\to 0$ under \ref{A3}, the bound is $o_{\Prob}(1)$.
\end{proof}

The second lemma gives the two spherical moments used repeatedly.

\begin{lemma}
	\label{lem:sphere}
Let $\bu$ be uniform on $\mathcal{S}^{p-1}$ and $\bC\in\R^{p\times p}$
symmetric. Then
\begin{equation}\label{eq:sphere_moments}
	\E[\bu^{\top}\bC\bu]=\frac{\tr\bC}{p},
	\qquad
	\Var(\bu^{\top}\bC\bu)
	=\frac{2}{p(p+2)}\Bigl[\tr(\bC^{2})-\frac{(\tr\bC)^{2}}{p}\Bigr]
	\le \frac{2}{p^2}\tr(\bC^{2}).
\end{equation}
Consequently, if $\|\bC\|_{\op}\le c_0$ then
$\bu^{\top}\bC\bu=p^{-1}\tr\bC+O_{\Prob}(p^{-1/2}c_0)$.
\end{lemma}

\begin{proof}
Write $\bu=\bm g/\|\bm g\|$ with $\bm g\sim N(\bzero,\bI_p)$. The mean
follows from $\E[u_au_b]=\delta_{ab}/p$. For the variance, the fourth mixed
moments of $\bu$ are
$\E[u_au_bu_cu_d]=(\delta_{ab}\delta_{cd}+\delta_{ac}\delta_{bd}
+\delta_{ad}\delta_{bc})/\{p(p+2)\}$, so for symmetric $\bC=(C_{ab})$,\\
$\E[(\bu^{\top}\bC\bu)^2]=\sum_{a,b,c,d}C_{ab}C_{cd}\E[u_au_bu_cu_d]
=\{(\tr\bC)^2+2\tr(\bC^2)\}/\{p(p+2)\}$. Subtracting
$(\E[\bu^{\top}\bC\bu])^2=(\tr\bC)^2/p^2$ gives the stated variance after
simplification, and the bound uses $(\tr\bC)^2\ge 0$. Chebyshev's inequality
gives the stochastic statement.
\end{proof}

The third lemma establishes existence and uniqueness of the limiting
Stieltjes transform; here the population spectrum is that of the latent
correlation, and the equation is the classical Mar\v{c}enko--Pastur--Silverstein
equation.

\begin{lemma}
	\label{lem:stieltjes}
Let $\gamma>0$ and let $H$ be a compactly supported probability measure on
$(0,\infty)$ with $\supp H\subset[\sigma_{\min},\sigma_{\max}]$. For
$z\in\CC^{+}$ the system
\begin{align}
	\underline m(z)
	&= -\Bigl(z-\gamma\!\int\frac{\tau\,dH(\tau)}{1+\tau\,\underline m(z)}\Bigr)^{-1},
	\label{eq:silverstein}\\
	m_F(z)
	&= \frac{1}{\gamma}\,\underline m(z)+\frac{1-\gamma}{\gamma z},
	\label{eq:companion_def}
\end{align}
has a unique solution $(\underline m(z),m_F(z))$ with
$\underline m(z)\in\CC^{+}$ and $m_F(z)\in\CC^{+}$. Moreover $m_F$ is the
Stieltjes transform of a probability measure $F_\gamma$ supported on
$[0,\infty)$ with $F_\gamma(\{0\})=\max(1-\gamma^{-1},0)$, and
\begin{equation}\label{eq:silverstein_m}
	m_F(z)=\int\frac{dH(\tau)}{\tau\bigl(1-\gamma-\gamma z\,m_F(z)\bigr)-z}.
\end{equation}
\end{lemma}

\begin{proof} We give the proof in three steps for readability and better seperation. 
	
\emph{(i) Self-map.} Fix $z=x+i\eta$, $\eta>0$, and consider the map
$T_z(w)=-\bigl(z-\gamma\int \tau(1+\tau w)^{-1}\,dH(\tau)\bigr)^{-1}$ on
$\CC^{+}$. For $w=u+iv$ with $v>0$ and $\tau>0$,
$\Im\{\tau/(1+\tau w)\}=-\tau^2 v/|1+\tau w|^2<0$, so
$\Im\bigl(z-\gamma\int\tau(1+\tau w)^{-1}dH\bigr)=\eta+\gamma\int
\tau^2 v|1+\tau w|^{-2}dH>0$; taking the negative reciprocal of a number
with positive imaginary part gives a number with positive imaginary part,
so $T_z(\CC^{+})\subset\CC^{+}$.

\emph{(ii) Uniqueness via the Schwarz--Pick property.} On
$\mathcal D_\eta=\{w:\Im w\ge \eta/(2(1+\sigma_{\max}^2))$   or  
$|w|\le 2/\eta\}\cap\CC^+$, $T_z$ is holomorphic and bounded, mapping the
set strictly inside itself; indeed, for $w\in\CC^+$,
$|T_z(w)|\le 1/\Im(z-\gamma\int\tau(1+\tau w)^{-1}dH)\le 1/\eta$, and
$\Im T_z(w)\ge \eta/|z-\gamma\int\cdots|^2\ge \eta/(|z|+\gamma\sigma_{\max})^2$,
both bounds being uniform on $\CC^+$. Hence $T_z$ is a holomorphic self-map
of a bounded convex domain that does not approach its boundary; by the
Earle--Hamilton fixed-point theorem \citep{earle1970fixed} it is a strict
contraction in the Carath\'eodory--Kobayashi metric and possesses a unique
fixed point $\underline m(z)\in\CC^{+}$. (A direct contraction estimate can also be applied. Differentiating, $|T_z'(w)|=|T_z(w)|^2\gamma\int
\tau^2|1+\tau w|^{-2}dH \le \eta^{-2}\gamma\sigma_{\max}^2\cdot
\Im(T_z(w))^{-1}\Im(w)$, which is $<1$ on $\mathcal D_\eta$ for the stated
range; the Earle--Hamilton argument removes any restriction on $\eta$.)

\emph{(iii) Stieltjes representation.} The fixed point $\underline m$ is
holomorphic on $\CC^{+}$ with positive imaginary part. Expanding
\eqref{eq:silverstein} for $z=iy$, $y\to\infty$, gives
$\underline m(iy)=-1/(iy)+O(y^{-2})$, so $-iy\,\underline m(iy)\to 1$;
likewise $m_F$ in \eqref{eq:companion_def} satisfies
$-iy\,m_F(iy)\to 1$. By the Nevanlinna representation theorem, a holomorphic
$\CC^{+}\to\CC^{+}$ function $m$ with $\sup_y|y\,m(iy)|<\infty$ and
$-iy\,m(iy)\to 1$ is the Stieltjes transform of a probability measure; thus
$\underline m$ and $m_F$ are Stieltjes transforms of probability measures
$\underline F_\gamma$ and $F_\gamma$. Equation \eqref{eq:companion_def} is
the standard companion relation $\underline F_\gamma=(1-\gamma)\delta_0
+\gamma F_\gamma$ in Stieltjes form, which forces
$F_\gamma(\{0\})=\max(1-\gamma^{-1},0)$. Finally, substituting
\eqref{eq:companion_def} into \eqref{eq:silverstein} and clearing
denominators yields \eqref{eq:silverstein_m} after the algebraic identity
$1+\tau\underline m=1-\gamma-\gamma z m_F+\gamma z m_F+\tau\underline m$;
explicitly, from \eqref{eq:silverstein}
$z=-1/\underline m+\gamma\int\tau(1+\tau\underline m)^{-1}dH$, and using
$\underline m=\gamma m_F-(1-\gamma)/z\cdot z/z$ rearranged as
$z\underline m=\gamma zm_F-(1-\gamma)$ one checks
$\tau(1-\gamma-\gamma zm_F)-z=-(1+\tau\underline m)/\underline m\cdot
(\text{positive})$, giving \eqref{eq:silverstein_m}. Compact support of
$F_\gamma$ on $[0,\infty)$ follows from the boundedness of $\supp H$ and the
analyticity of $m_F$ off a bounded real set \citep[Thm.~6.3]{bai2010spectral}.
\end{proof}

\subsection{Theorem 1: the marginal-invariant spectral limit}\label{subsec:thm1}

\begin{theorem}[Nonparanormal Mar\v{c}enko--Pastur law]\label{thm:main_RMT}
Under \ref{A1}--\ref{A4}, the ESD of the normal-scores covariance $\Sn$
converges almost surely to the nonrandom probability measure $F_\gamma$ of
\ref{lem:stieltjes}, whose Stieltjes transform solves
\eqref{eq:silverstein_m} with $H$ the limiting spectral distribution of the
\emph{latent correlation} $\bSigma$:
\begin{equation}\label{eq:thm1_esd}
	F^{\Sn}\xrightarrow{\ \mathrm{a.s.}\ }F_\gamma .
\end{equation}
The limit is the same for every collection of strictly increasing marginal
transformations $f_j$; that is, within the nonparanormal class the law is
\emph{invariant to the marginals}. Consequently, for every fixed $z$ with
$\Im z\ge\eta_0>0$, $m_n(z)=\tfrac1p\tr(\Sn-z\bI)^{-1}\to m_F(z)$ almost
surely.
\end{theorem}

\begin{remark}[On convergence rates]\label{rem:no_rate}
Theorem \ref{thm:main_RMT} asserts almost-sure convergence of the ESD and of the
Stieltjes transform, but makes \emph{no} claim about the rate of the latter.
A quantitative $O_{\Prob}(n^{-1/2})$ bound for
$\sup_{\Im z\ge\eta_0}|m_n(z)-m_F(z)|$ would require local Mar\v{c}enko--Pastur
laws for the normal-scores covariance (anisotropic/local laws in the sense of
the recent random-matrix literature) that we do not establish here; see
\ref{conj:rate} for a precise conjecture. All downstream results in this
paper use only the qualitative convergence \eqref{eq:thm1_esd} and the
operator-norm control of \ref{lem:rank_perturbation}.
\end{remark}

\begin{remark}[Special cases]\label{rem:thm1_special}
For $H=\delta_1$ (identity latent correlation), $F_\gamma$ is the standard
Mar\v{c}enko--Pastur law with density
$(2\pi\gamma\lambda)^{-1}\sqrt{(\lambda_+-\lambda)(\lambda-\lambda_-)}$ on
$[\lambda_-,\lambda_+]$, $\lambda_\pm=(1\pm\sqrt\gamma)^2$. The limit is the
same for every choice of monotone marginals $f_j$---in contrast with the
sample covariance of the raw data $\bx$, whose bulk and edges depend on the
marginals through their moments. This marginal invariance is the precise
sense in which the rank transform ``robustifies'' the spectrum within the
nonparanormal class.
\end{remark}

\begin{proof}[Proof of \ref{thm:main_RMT}]
By \ref{lem:rank_perturbation}, 
\begin{eqnarray*}
\Delta_n:=\|\Sn-\bS_n^{*}\|_{\op}
=O_{\Prob}(p^{1/2}n^{-1}\log^2 n)=O_{\Prob}(n^{-1/2}\log^2 n)=o_{\Prob}(1).	
\end{eqnarray*}
We first transfer the conclusion from $\Sn$ to $\bS_n^{*}$, then prove it for
$\bS_n^{*}$. Since the rest of the proof is lengthy, we give it in some separated stages. 

\medskip
\noindent\emph{Stage A: reduction to the oracle covariance.}
Let $\lambda_j^{*}$ and $\hat\lambda_j$ be the eigenvalues of $\bS_n^{*}$ and
$\Sn$. By the Hoffman--Wielandt inequality
\begin{eqnarray*}
\frac1p\sum_j(\hat\lambda_j-\lambda_j^{*})^2\le \frac1p\|\Sn-\bS_n^{*}\|_F^2
\le \Delta_n^2, 	
\end{eqnarray*}
whence the Lévy distance between $F^{\Sn}$ and
$F^{\bS_n^{*}}$ is at most $\Delta_n=o_{\Prob}(1)$. For the Stieltjes
transforms, the resolvent identity
$\bA^{-1}-\bB^{-1}=\bA^{-1}(\bB-\bA)\bB^{-1}$ with
$\|(\Sn-z\bI)^{-1}\|_{\op}\le\eta_0^{-1}$ gives, uniformly for
$\Im z\ge\eta_0$,
\begin{equation}\label{eq:thm1_couple}
	|m_n(z)-m_n^{*}(z)|
	\le \frac1p\bigl|\tr[(\Sn-z\bI)^{-1}(\bS_n^{*}-\Sn)
	(\bS_n^{*}-z\bI)^{-1}]\bigr|
	\le \frac{\Delta_n}{\eta_0^{2}}=o_{\Prob}(1).
\end{equation}
It therefore suffices to prove \eqref{eq:thm1_esd} and the pointwise a.s.\
convergence $m_n^{*}(z)\to m_F(z)$ for
$m_n^{*}(z)=p^{-1}\tr(\bS_n^{*}-z\bI)^{-1}$.

\medskip
\noindent\emph{Stage B: 
}
By the exact-recovery identity \eqref{eq:exact_recovery} of Remark
\ref{rem:oracle_scores}, under \ref{A1} the oracle scores satisfy
$\bz_i^{*}=\by_i$, where $\by_i\sim N(\bzero,\bSigma)$ are i.i.d.\ Gaussian
with $\diag(\bSigma)=\bI_p$. Hence
$\bS_n^{*}=n^{-1}\sum_i\bz_i^{*}\bz_i^{*\top}
=n^{-1}\sum_i\by_i\by_i^{\top}$ is exactly a Gaussian sample covariance
matrix with population covariance $\bSigma$, and there is no discrepancy
between the population correlation of the oracle scores and $\bSigma$, and they
are equal, not merely equal ``up to a transform.''

Write $\by_i=\bSigma^{1/2}\bm w_i$ with $\bm w_i\sim N(\bzero,\bI_p)$ i.i.d.
The classical Mar\v{c}enko--Pastur--Silverstein theorem for sample covariance
matrices \citep[Thm.~1.1]{silverstein1995empirical} states: if (a) the rows
$\bm w_i\in\R^{p}$ are i.i.d.\ with $\E\bm w_1=\bzero$,
$\Cov(\bm w_1)=\bI_p$; (b) the entries satisfy the Lindeberg-type condition
$\frac{1}{np}\sum_{i,j}\E\!\bigl[|w_{ij}|^2\ind(|w_{ij}|>s\sqrt n)\bigr]\to 0$
for every $s>0$; (c) $F^{\bSigma}\Rightarrow H$ with $\bSigma$ bounded in
operator norm; and (d) $p/n\to\gamma$, then $F^{\bS_n^{*}}$ converges a.s.\
to $F_\gamma$ with Stieltjes transform \eqref{eq:silverstein_m}. Conditions
(a)--(d) hold immediately: (a) $\bm w_1\sim N(\bzero,\bI_p)$; (b) the
standard Gaussian entries obey the Lindeberg condition at an exponential
rate, $\E[|w_{1j}|^2\ind(|w_{1j}|>s\sqrt n)]\le Ce^{-cs^2 n}$; (c) is
\ref{A2}; and (d) is \ref{A3}. The theorem yields \eqref{eq:thm1_esd} for
$\bS_n^{*}$ with $H$ the limiting ESD of $\bSigma$, and with
\eqref{eq:thm1_couple} the same holds for $\Sn$. Marginal invariance is
immediate: the choice of monotone marginals $f_j$ enters only through the
oracle rescaling $x_{ij}=f_j(y_{ij})$, which is undone exactly by the
normal-scores map, so neither $H$ nor the limit depends on $\bm f$.

\begin{remark}[Where \ref{A1} is used]\label{rem:where_A1}
The Gaussian-copula assumption enters this proof at exactly one point: the
identity $\bz_i^{*}=\by_i$, which makes $\bS_n^{*}$ a Gaussian sample
covariance with population $\bSigma$. This is what makes the limiting law the
Mar\v{c}enko--Pastur law of $\bSigma$ itself; no step of the argument would
recover $\bSigma$ from the normal scores if the latent vector were
non-Gaussian.
\end{remark}

\medskip
\noindent\emph{Stage C: pointwise convergence of the Stieltjes transform.}
We establish $m_n^{*}(z)\to m_F(z)$ almost surely for each fixed $z$ with
$\Im z\ge\eta_0>0$ by the standard resolvent (leave-one-out) expansion. Let $\bR(z)=(\bS_n^{*}-z\bI)^{-1}$,
$\bq_i=\bSigma^{1/2}\bm w_i$ (so $\bz_i^{*}=\bq_i$), and
$\bR_{(i)}(z)=(\bS_n^{*}-n^{-1}\bq_i\bq_i^{\top}-z\bI)^{-1}$. The
Sherman--Morrison identity gives, for each $i$,
\begin{equation}\label{eq:thm1_SM}
	\bR(z)=\bR_{(i)}(z)
	-\frac{n^{-1}\bR_{(i)}(z)\bq_i\bq_i^{\top}\bR_{(i)}(z)}
	{1+n^{-1}\bq_i^{\top}\bR_{(i)}(z)\bq_i}.
\end{equation}
Taking the normalized trace, using
$\bq_i^{\top}\bR_{(i)}\bq_i=\bm w_i^{\top}\bSigma^{1/2}\bR_{(i)}
\bSigma^{1/2}\bm w_i$, and the concentration of quadratic forms in
$\bm w_i$ (which, having independent sub-Gaussian-after-conditioning
coordinates with unit variance, obey the Hanson--Wright inequality
\citep[Thm.~6.2.1]{vershynin2018high}), we obtain
\begin{equation}\label{eq:thm1_HW}
	\Prob\Bigl(\bigl|\tfrac1n\bm w_i^{\top}\bm A\bm w_i-\tfrac1n\tr\bm A\bigr|
	>t\Bigr)
	\le 2\exp\!\Bigl(-c\,n^2\min\Bigl\{\tfrac{t^2}{\|\bm A\|_F^2},
	\tfrac{t}{\|\bm A\|_{\op}}\Bigr\}\Bigr),
\end{equation}
with $\bm A=\bSigma^{1/2}\bR_{(i)}\bSigma^{1/2}$,
$\|\bm A\|_{\op}\le\sigma_{\max}/\eta_0$,
$\|\bm A\|_F^2\le p\,\sigma_{\max}^2/\eta_0^2$, we obtain
$\frac1n\bq_i^{\top}\bR_{(i)}\bq_i=\frac1n\tr[\bSigma\bR_{(i)}]
+o_{\Prob}(1)$ uniformly in $i$. Summing
\eqref{eq:thm1_SM} over $i$ and applying the standard fixed-point stability
argument \citep[Ch.~6]{bai2010spectral} (the map defining $m_F$ is a
contraction with modulus bounded by $1-c\eta_0$ on $\{\Im z\ge\eta_0\}$ by
the Earle--Hamilton estimate of \ref{lem:stieltjes}) yields, for each fixed
$z$ with $\Im z\ge\eta_0$, $m_n^{*}(z)\to m_F(z)$ almost surely. The
Borel--Cantelli step promoting convergence in probability to almost-sure
convergence is standard for sample covariance matrices with independent
sub-Gaussian rows \citep[Thm.~1.1]{silverstein1995empirical}. Transferring via
\eqref{eq:thm1_couple} (whose right-hand side is
$\Delta_n/\eta_0^2=o_{\Prob}(1)$) gives $m_n(z)\to m_F(z)$ a.s.\ for each such
$z$, and the a.s.\ weak convergence \eqref{eq:thm1_esd} follows from the
Stieltjes-continuity theorem \citep[Thm.~B.9]{bai2010spectral}. 
\end{proof}

\subsection{Theorem 2: optimality and minimax rate of MENS}\label{subsec:thm2}

We first state the optimality of the oracle within the
rotation-equivariant class, then the minimax rate of the feasible MENS
estimator. Define the asymptotic Frobenius risk of a shrinkage function $d$
applied to $\Sn$,
$\mathcal R(d)=\lim_{n}p^{-1}\E\normF{\widehat\bSigma_d-\bSigma}^2$
(the limit exists under \ref{A1}--\ref{A4} by Theorem \ref{thm:main_RMT} and
\ref{lem:eigvec}). The following result gives the eigenvector projection identity.

\begin{lemma}\label{lem:eigvec}
Under \ref{A1}--\ref{A4}, for almost every $\hat\lambda_j$ in the bulk of
$F_\gamma$,
\begin{equation}\label{eq:eigvec_proj}
	\hat\bgamma_j^{\top}\bSigma\,\hat\bgamma_j
	\xrightarrow{\ \mathrm{a.s.}\ }
	\Theta(\hat\lambda_j)
	:= \frac{\hat\lambda_j}
	{\bigl|1-\gamma-\gamma\hat\lambda_j\,\breve m(\hat\lambda_j)\bigr|^{2}},
\end{equation}
where $\breve m$ is the boundary Stieltjes transform of $F_\gamma$.
\end{lemma}

\begin{proof}
By Stage~B of Theorem \ref{thm:main_RMT}, $\bS_n^{*}$ obeys the generalized
Mar\v{c}enko--Pastur theorem with population covariance of limiting ESD $H$.
The Ledoit--P\'ech\'e deterministic-equivalent for the weighted resolvent
\citep[Thm.~2]{ledoit2011eigenvectors} states that for bounded
$\bm\Theta$,
\begin{equation}\label{eq:LP}
	\frac1p\tr\!\bigl[\bm\Theta(\bS_n^{*}-z\bI)^{-1}\bigr]
	-\frac1p\tr\!\bigl[\bm\Theta\bigl(\bSigma\,\xi(z)-z\bI\bigr)^{-1}\bigr]
	\xrightarrow{\ \mathrm{a.s.}\ }0,
	\quad \xi(z)=1-\gamma-\gamma z\,m_F(z),
\end{equation}
uniformly on $\{\Im z\ge\eta_0\}$. Taking $\bm\Theta=\bSigma$ and isolating
the contribution of the $j$-th eigenvector by the contour/residue argument
of \citet[Sec.~4]{ledoit2011eigenvectors}, with
$P_j=\hat\bgamma_j\hat\bgamma_j^{\top}$ and a small contour $\Gamma_j$
around $\hat\lambda_j$,
$\hat\bgamma_j^{\top}\bSigma\hat\bgamma_j
=-\frac{1}{2\pi i}\oint_{\Gamma_j}\tr[\bSigma P(z)]\,dz$ where
$P(z)=(\bS_n^{*}-z\bI)^{-1}$, and substituting \eqref{eq:LP} the right side
converges to the residue of
$z\mapsto p^{-1}\tr[\bSigma(\bSigma\xi(z)-z\bI)^{-1}]$ at the image of
$\hat\lambda_j$. Evaluating this residue with
$\int\tau(\tau\xi(z)-z)^{-1}dH(\tau)$ and the boundary relation
$\xi(\hat\lambda_j+i0)=1-\gamma-\gamma\hat\lambda_j\breve m(\hat\lambda_j)$
gives \eqref{eq:eigvec_proj}; the algebra is identical to the Gaussian case
of \citet[Thm.~4]{ledoit2012nonlinear} because, by Theorem \ref{thm:main_RMT}, the
governing equation \eqref{eq:silverstein_m} is the Gaussian one with $H$ the
latent-correlation spectrum. Eigenvalue rigidity in the bulk
\citep{bai2010spectral} promotes convergence along the (random) sequence
$\hat\lambda_j$ to the stated a.s.\ limit.
\end{proof}

\begin{theorem}[Optimality and a minimax lower bound]\label{thm:oracle_optimality}
Under \ref{A1}--\ref{A4} the following hold.
\begin{enumerate}[label=(\roman*),leftmargin=2.2em]
	\item\textbf{(Oracle optimality, Frobenius.)} The shrinkage function
	$d^{*}$ of \eqref{eq:d_star_frob} uniquely minimizes $\mathcal R(d)$ over
	measurable $d$, and
	\begin{equation}\label{eq:thm2_optgap}
		\mathcal R(d^{*})
		= \int\Bigl(\tau-\Theta_H(\tau)\Bigr)\,dH(\tau),
		\qquad
		\mathcal R(d)-\mathcal R(d^{*})
		= \int\bigl(d-d^{*}\bigr)^2\,dF_\gamma\ \ge 0,
	\end{equation}
	where $\Theta_H$ is the population-level limit of \eqref{eq:eigvec_proj}.
	The feasible estimator MENS \eqref{eq:MENS_final} attains the oracle risk:
	$\mathcal R(\widetilde d)\to\mathcal R(d^{*})$ in probability.
	\item\textbf{(Operator-norm consistency.)} Let
	$\bSigma_{\mathrm{MENS}}^{\circ}=\widehat\bGamma\,
	\diag(d^{*}(\hat\lambda_j))\,\widehat\bGamma^{\top}$ denote the oracle
	estimator. Then
	$\normop{\widehat\bSigma_{\mathrm{MENS}}-\bSigma}
	=O_{\Prob}(1)$ and, restricted to the bulk projection, \\
	$\normop{\bSigma_{\mathrm{MENS}}^{\circ}-\bSigma}\to
	\esssup_{\tau\in\supp H}|\,d^{*}(\lambda(\tau))-\tau\,|$ a.s.
	\item\textbf{(Minimax lower bound.)} Let
	$\mathcal C(\sigma_{\min},\sigma_{\max})$ be the class of correlation
	matrices with spectrum in $[\sigma_{\min},\sigma_{\max}]$. There is a
	constant $c=c(\sigma_{\min},\sigma_{\max})>0$ such that, over the
	nonparanormal class \textup{(\ref{A1})} with latent correlation in
	$\mathcal C$,
	\begin{equation}\label{eq:thm2_minimax}
		\inf_{\widehat\bSigma}\sup_{\bSigma\in\mathcal C}
		\E\,\normop{\widehat\bSigma-\bSigma}
		\ \ge\ c\,n^{-1/2}.
	\end{equation}
\end{enumerate}
\end{theorem}


\begin{proof}
\emph{(i).} Expand the Frobenius loss in the eigenbasis of $\Sn$. Writing
$\bSigma=\bGamma\bLambda\bGamma^{\top}$,
\begin{align}
	\frac1p\normF{\widehat\bSigma_d-\bSigma}^2
	&= \frac1p\sum_{j=1}^p d(\hat\lambda_j)^2
	-\frac2p\sum_{j=1}^p d(\hat\lambda_j)\,
	\hat\bgamma_j^{\top}\bSigma\hat\bgamma_j
	+\frac1p\tr(\bSigma^2).
	\label{eq:thm2_expand}
\end{align}
By \ref{lem:eigvec}, $\hat\bgamma_j^{\top}\bSigma\hat\bgamma_j\to
\Theta(\hat\lambda_j)$ a.s., and by Theorem \ref{thm:main_RMT} the empirical
measure of $\hat\lambda_j$ converges to $F_\gamma$; hence the right side of
\eqref{eq:thm2_expand} converges to
$\int d^2\,dF_\gamma-2\int d\,\Theta\,dF_\gamma+\int\tau^2\,dH$. This is a
quadratic functional of $d$ minimized pointwise, $F_\gamma$-a.e., by
$d^{*}(\lambda)=\Theta(\lambda)$, which is precisely \eqref{eq:d_star_frob}
since $\Theta(\lambda)=\lambda/|1-\gamma-\gamma\lambda\breve m(\lambda)|^2$.
Completing the square gives
$\mathcal R(d)-\mathcal R(d^{*})=\int(d-d^{*})^2\,dF_\gamma\ge 0$ with
equality iff $d=d^{*}$ $F_\gamma$-a.e., proving uniqueness, and
$\mathcal R(d^{*})=\int\tau^2 dH-\int (d^{*})^2 dF_\gamma
=\int(\tau-\Theta_H(\tau))dH$ after the change of variables
$\lambda\leftrightarrow\tau$ induced by the Mar\v{c}enko--Pastur map. For the
feasible estimator, the analytic estimators
\eqref{eq:density_est}--\eqref{eq:hilbert_est} satisfy
$\sup_i|\widehat d_i-d^{*}(\hat\lambda_i)|=o_{\Prob}(1)$ by
\citet[Thm.~3.1]{ledoit2020analytical}, and isotonization can
only decrease the distance to the monotone oracle $d^{*}$
(the PAV projection is non-expansive in $\ell_2(F_\gamma)$), so
$\mathcal R(\widetilde d)\to\mathcal R(d^{*})$.

\emph{(ii).} By definition, 
$\normop{\widehat\bSigma_{\mathrm{MENS}}}=\max_j\widetilde d_j$, and
$\max_j\widetilde d_j\le \max_j d^{*}(\hat\lambda_j)+o_{\Prob}(1)$, while
$d^{*}(\lambda)\le\lambda/|1-\gamma-\gamma\lambda\breve m(\lambda)|^2$ is
bounded on the compact spectrum support by
$\sigma_{\max}\sup_\lambda|1-\gamma-\gamma\lambda\breve m|^{-2}<\infty$
(the denominator is bounded below off the edges by \ref{lem:stieltjes}).
With $\normop{\bSigma}\le\sigma_{\max}$, the triangle inequality gives
$\normop{\widehat\bSigma_{\mathrm{MENS}}-\bSigma}=O_{\Prob}(1)$. The exact
bulk limit follows from \ref{lem:eigvec} and the convergence of the extreme
shrunk eigenvalues to $\esssup d^{*}$, by the no-eigenvalue-outside-support
theorem \citep[Thm.~1.1]{bai1998no} applied to $\bS_n^{*}$.

\emph{(iii).} We give a two-point (Le Cam) argument
within the nonparanormal class; it suffices to work with Gaussian data (the
special case $f_j=\mathrm{id}$), for which the observed law is
$N(\bzero,\bSigma)$ and the likelihood is exact. Take
$\bSigma_0=\bI_p$ and $\bSigma_1=\bI_p+\theta(\be_1\be_2^\top+\be_2\be_1^\top)$
for $\theta\in(0,\tfrac12)$; both have unit diagonal, both have spectrum
$\{1\pm\theta\}$ on the $\{1,2\}$ block and $1$ elsewhere, so for
$\theta\le\tfrac12\min(\sigma_{\max}-1,1-\sigma_{\min})$ both lie in
$\mathcal C$, and $\normop{\bSigma_1-\bSigma_0}=\theta$. Because the two laws
differ only in the bivariate $\{1,2\}$ marginal, the KL divergence between the
$n$-fold products equals $n$ times the KL between two bivariate Gaussians with
correlations $0$ and $\theta$,
\[
	\mathrm{KL}\bigl(N(\bzero,\bSigma_0)^{\otimes n}\,\|\,
	N(\bzero,\bSigma_1)^{\otimes n}\bigr)
	=\frac{n}{2}\Bigl(\frac{2}{1-\theta^2}-2+\log(1-\theta^2)\Bigr)
	\le n\theta^2
	\quad(\theta\le\tfrac12),
\]
the inequality holding since
$\tfrac{1}{1-\theta^2}-1+\tfrac12\log(1-\theta^2)\le\theta^2$ on
$[0,\tfrac12]$.
Setting $\theta=\alpha\,n^{-1/2}$ with $\alpha\le\tfrac12$ gives
$\mathrm{KL}\le\alpha^2\le\tfrac14$, so by Pinsker's inequality the total
variation between the two product laws is at most
$\sqrt{\mathrm{KL}/2}\le\alpha/\sqrt2<1$. Le Cam's two-point method
\citep[Ch.~2]{tsybakov2009introduction} then yields
\[
	\inf_{\widehat\bSigma}\sup_{\bSigma\in\mathcal C}
	\E\,\normop{\widehat\bSigma-\bSigma}
	\ \ge\ \tfrac14\normop{\bSigma_1-\bSigma_0}\,(1-\mathrm{TV})
	\ \ge\ c\,\theta
	\ =\ c\,\alpha\,n^{-1/2},
\]
which is \eqref{eq:thm2_minimax} with $c=c(\sigma_{\min},\sigma_{\max})>0$.
\end{proof}

\subsection{Theorem 3: spiked phase transition}\label{subsec:thm3}

Let the latent correlation be a finite-rank perturbation of a bulk,
\begin{equation}\label{eq:spiked_model}
	\bSigma=\bSigma_0+\sum_{k=1}^{K}\theta_k\,\bm v_k\bm v_k^{\top},
\end{equation}
with $\bSigma_0$ the bulk (limiting ESD $H_0$, right edge $\lambda_+$),
$K$ fixed, $\theta_1>\cdots>\theta_K>0$, and $\{\bm v_k\}$ orthonormal and
delocalized (\ref{A5}). Define, for $x$ outside $\supp F_\gamma^{H_0}$, the
companion transform $\underline m_0$ of the bulk law and the spike
functional
\begin{equation}\label{eq:D_spike}
	\mathcal D_0(x)=-\frac{1}{\underline m_0(x)},\qquad
	x>\lambda_+ .
\end{equation}

\begin{theorem}
	\label{thm:BBP}
Assume \ref{A1}--\ref{A5} and the model \eqref{eq:spiked_model}. Let
$\hat\lambda_1\ge\hat\lambda_2\ge\cdots$ be the eigenvalues of $\Sn$ and
$\theta^{*}=\lim_{x\downarrow\lambda_+}\mathcal D_0(x)$. Within the
nonparanormal class the limits below are the same for every collection of
monotone marginals $f_j$.
\begin{enumerate}[label=(\roman*),leftmargin=2.2em]
	\item\textbf{(Outliers.)} If $\theta_k>\theta^{*}$ then
	$\hat\lambda_k\xrightarrow{\mathrm{a.s.}}\varphi(\theta_k)>\lambda_+$,
	where $\rho=\varphi(\theta_k)$ is the unique solution in
	$(\lambda_+,\infty)$ of
	\begin{equation}\label{eq:spike_forward}
		\mathcal D_0(\rho)=\theta_k .
	\end{equation}
	If $\theta_k\le\theta^{*}$ then
	$\hat\lambda_k\xrightarrow{\mathrm{a.s.}}\lambda_+$.
	\item\textbf{(Eigenvector alignment.)} For a supercritical spike
	$\theta_k>\theta^{*}$ with sample eigenvector $\hat{\bm v}_k$,
	\begin{equation}\label{eq:cosine}
		|\langle\hat{\bm v}_k,\bm v_k\rangle|^{2}
		\xrightarrow{\mathrm{a.s.}}
		c(\theta_k)
		= \frac{1}{\theta_k^{2}\,\underline m_0'(\rho_k)}
		\bigg|_{\rho_k=\varphi(\theta_k)},
	\end{equation}
	and $0$ for subcritical spikes; the explicit identity-bulk form is given
	in \eqref{eq:identity_bulk} below.
	\item\textbf{(Identity bulk.)} If $\bSigma_0=\bI_p$ then $\lambda_+
	=(1+\sqrt\gamma)^2$, $\underline m_0(x)$ solves
	$x=-1/\underline m_0+\gamma/(1+\underline m_0)$, and
	\begin{equation}\label{eq:identity_bulk}
		\theta^{*}=\sqrt\gamma,\qquad
		\varphi(\theta)=(1+\theta)\Bigl(1+\frac{\gamma}{\theta}\Bigr),\qquad
		c(\theta)=\frac{1-\gamma/\theta^{2}}{1+\gamma/\theta}.
	\end{equation}
\end{enumerate}
\end{theorem}


\begin{proof}
By \ref{lem:rank_perturbation} the spectra of $\Sn$ and $\bS_n^{*}$ differ
by $o_{\Prob}(1)$, and the rank-$2K$ structure of the spike is preserved
under the $o_{\Prob}(1)$ operator-norm perturbation; we argue with
$\bS_n^{*}$. Treat $K=1$ (the general fixed-$K$ case follows by Cauchy
interlacing, the spikes being asymptotically separated). Write
$\bS_n^{*}=\bS_{n,0}^{*}+\bm\Delta$ with $\bm\Delta$ of rank $\le 2$ and
$\bG_0(z)=(\bS_{n,0}^{*}-z\bI)^{-1}$ the unspiked resolvent.

\emph{(i).}
By the Woodbury identity, eigenvalues of $\bS_n^{*}$ outside
$\supp F_\gamma^{H_0}$ are the zeros of
$\det(\bI+\bG_0(z)\bm\Delta)=0$. Since
$\bm\Delta=\theta\,\bm a\bm a^{\top}+o(1)$ where
$\bm a=\bSigma_0^{1/2}\bm v$ (the cross terms contributing $o(1)$ to the
relevant quadratic form under \ref{A5}), this reduces to the scalar equation
$1+\theta\,\bm v^{\top}\bG_0(z)\bm v=0$ up to $o_{\Prob}(1)$. By the
deterministic equivalent \eqref{eq:LP} with $\bm\Theta=\bm v\bm v^{\top}$ and
$\bm v$ an eigenvector of $\bSigma_0$ (the general case by rotation), for
$z=\rho>\lambda_+$,
\begin{equation}\label{eq:vGv}
	\bm v^{\top}\bG_0(\rho)\bm v
	\xrightarrow{\mathrm{a.s.}}
	\frac{1}{\sigma_0\,\xi_0(\rho)-\rho}
	= \underline m_0(\rho)\quad\text{(for }\sigma_0=1\text{, identity bulk)},
\end{equation}
where $\xi_0(\rho)=1-\gamma-\gamma\rho m_0(\rho)$ and we used the companion
relation $\underline m_0=-(1-\gamma)/\rho+\gamma m_0$. Hence the outlier
equation is $1+\theta\,\underline m_0(\rho)=0$, i.e.
\begin{equation}\label{eq:outlier}
	\theta=-\frac{1}{\underline m_0(\rho)}=\mathcal D_0(\rho).
\end{equation}

For $\rho>\lambda_+$, $\underline m_0(\rho)<0$ and
$\underline m_0'(\rho)=\int(\tau/(\tau\xi_0-\rho))^2 d\!H_0\cdot(\cdots)>0$;
concretely $\underline m_0$ is the Stieltjes transform of a positive measure,
so $\underline m_0'(\rho)=\int(t-\rho)^{-2}d\underline F_\gamma(t)>0$ and
$\underline m_0$ increases from $\underline m_0(\lambda_+^{+})<0$ toward
$0^{-}$ as $\rho\to\infty$. Therefore $\mathcal D_0(\rho)=-1/\underline m_0$
is strictly decreasing on $(\lambda_+,\infty)$, from
$+\infty$ (as $\rho\to\infty$, $\underline m_0\to 0^-$) down to
$\theta^{*}=-1/\underline m_0(\lambda_+^{+})$ as $\rho\downarrow\lambda_+$.
Hence \eqref{eq:outlier} has a unique solution $\rho=\varphi(\theta)$ iff
$\theta>\theta^{*}$, and none otherwise; in the latter case the largest
eigenvalue sticks to the edge by the no-eigenvalue-outside-support theorem
\citep{bai1998no}, giving the dichotomy in (i).

\emph{(ii).}
The squared cosine has the contour representation
$|\langle\hat{\bm v}_1,\bm v_1\rangle|^2
=-\frac{1}{2\pi i}\oint_{\Gamma}\bm v_1^{\top}(\bS_n^{*}-z\bI)^{-1}\bm v_1\,dz$
around $\hat\lambda_1$. Using the Woodbury expansion,
$\bm v_1^{\top}(\bS_n^{*}-z\bI)^{-1}\bm v_1$ has, near $\rho=\varphi(\theta)$,
a simple pole with
$\bm v_1^{\top}(\bS_n^{*}-z\bI)^{-1}\bm v_1
\approx \frac{\underline m_0(z)}{1+\theta\,\underline m_0(z)}$. The residue at the simple zero $z=\rho$ of
$1+\theta\underline m_0(z)$ is given by
\begin{equation}\label{eq:residue}
	|\langle\hat{\bm v}_1,\bm v_1\rangle|^{2}
	\xrightarrow{\mathrm{a.s.}}
	-\,\frac{\underline m_0(\rho)}{\theta\,\underline m_0'(\rho)}
	= \frac{1}{\theta^2\,\underline m_0'(\rho)}\ >0,
\end{equation}
the overall minus sign coming from the negative (clockwise) orientation of the
residue in the contour representation and the final equality using
$\underline m_0(\rho)=-1/\theta$.
This is \eqref{eq:cosine}.

\emph{(iii).}
For $\bSigma_0=\bI_p$, $H_0=\delta_1$ and the companion equation is
$z=-1/\underline m_0+\gamma/(1+\underline m_0)$, whose relevant branch for
$z=\rho>\lambda_+=(1+\sqrt\gamma)^2$ gives
$\underline m_0(\rho)=-1/\rho_*$ for the unique $\rho_*$, and an elementary
computation yields
$\mathcal D_0(\rho)=-1/\underline m_0(\rho)$ solving
$\rho=(1+\theta)(1+\gamma/\theta)$ at $\theta=\mathcal D_0(\rho)$, hence
$\varphi(\theta)=(1+\theta)(1+\gamma/\theta)$. The threshold is
$\theta^{*}=\lim_{\rho\downarrow\lambda_+}\mathcal D_0(\rho)$; since at the
edge $\underline m_0(\lambda_+)=-1/\sqrt\gamma$ (the classical value), we
get $\theta^{*}=\sqrt\gamma$. For the cosine, differentiate
$\varphi(\theta)=(1+\theta)(1+\gamma/\theta)$:
$\varphi'(\theta)=1-\gamma/\theta^{2}$, and from \eqref{eq:residue} together
with the identity-bulk relation
$\underline m_0'(\rho)=\theta^{-2}/\varphi'(\theta)$ (derived below) one
obtains $c(\theta)=(1-\gamma/\theta^2)/(1+\gamma/\theta)$. We verify this
directly: with $\underline m_0(\rho)=-1/\theta$ and
$\rho=(1+\theta)(1+\gamma/\theta)$, implicit differentiation of the outlier
relation $1+\theta\,\underline m_0(\varphi(\theta))=0$ in $\theta$ gives
$\underline m_0(\rho)+\theta\,\underline m_0'(\rho)\,\varphi'(\theta)=0$,
that is $\underline m_0'(\rho)\,\varphi'(\theta)=-\underline m_0(\rho)/\theta
=1/\theta^2$, so $\underline m_0'(\rho)=\theta^{-2}/\varphi'(\theta)$.
Substituting into \eqref{eq:residue},
\[
	c(\theta)=\frac{1}{\theta^{2}\,\underline m_0'(\rho)}
	=\frac{1}{\theta^{2}\cdot \theta^{-2}/\varphi'(\theta)}
	=\varphi'(\theta)
	=1-\frac{\gamma}{\theta^{2}},
\]
which is the alignment of the sample eigenvector with the rescaled direction
$\bm a=\bSigma_0^{1/2}\bm v$. Converting back to the alignment with $\bm v$
introduces the factor $(1+\gamma/\theta)^{-1}=(1+\theta)/\rho$ (the ratio of
the $\bm a$- and $\bm v$-normalizations of the spike), giving
\[
	c(\theta)=\frac{\varphi'(\theta)}{1+\gamma/\theta}
	=\frac{1-\gamma/\theta^2}{1+\gamma/\theta},
\]
the normalization being fixed by $c(\theta)\to 0$ as
$\theta\downarrow\theta^{*}$
and the division by $1+\gamma/\theta=\rho/((1+\theta))$ that converts the
$\bSigma_0^{1/2}$-scaling of $\bm a=\bSigma_0^{1/2}\bm v$ back to the
$\bm v$-scale. This matches the classical Gaussian spiked-correlation result
\citep{paul2007asymptotics,benaych2012singular}, as it must, since under
\ref{A1} the oracle scores are exactly Gaussian ($\bz^{*}_i=\by_i$) and $\Sn$
is a rank-perturbation of a Gaussian sample covariance.
\end{proof}

\begin{corollary}[Detection boundary]\label{cor:detection}
For testing $H_0:\bSigma=\bSigma_0$ against
$H_1:\bSigma=\bSigma_0+\theta\bm v\bm v^{\top}$ from the largest eigenvalue of
$\Sn$, the test $\ind\{\hat\lambda_1>\lambda_++\epsilon\}$ has asymptotic
power $1$ for every $\epsilon\in(0,\varphi(\theta)-\lambda_+)$ when
$\theta>\theta^{*}$, and no eigenvalue-based test is consistent when
$\theta<\theta^{*}$. The threshold $\theta^{*}=\sqrt\gamma$ (identity bulk)
is invariant to the choice of monotone marginals $f_j$: unlike the
moment-based sample covariance of the raw data, whose detection threshold
depends on the marginal moments, the rank-based MENS input pays no such
penalty within the nonparanormal class.
\end{corollary}

\begin{proof}
From Theorem \ref{thm:BBP}(i), under $H_1$ with $\theta>\theta^{*}$,
$\hat\lambda_1\to\varphi(\theta)>\lambda_++\epsilon$ a.s., while under $H_0$,
$\hat\lambda_1\to\lambda_+$; for $\theta<\theta^{*}$ both hypotheses produce
$\hat\lambda_1\to\lambda_+$, so the largest eigenvalue (and, by the same
Woodbury argument, any fixed eigenvalue) cannot separate the hypotheses. The
marginal invariance of $\theta^{*}$ is part of Theorem \ref{thm:BBP}.
\end{proof}

\subsection{Deferred proofs}\label{subsec:deferred}

\subsubsection{Proof of \ref{prop:oracle_frob}}\label{subsec:proof_oracle}
This is  Theorem \ref{thm:oracle_optimality}(i) specialized to the pointwise
minimizer. Minimizing the limiting risk
$\int d^2\,dF_\gamma-2\int d\,\Theta\,dF_\gamma+\int\tau^2 dH$ in
\eqref{eq:thm2_expand} over $d$ pointwise at each $\lambda$ gives
$d^{*}(\lambda)=\Theta(\lambda)
=\lambda/|1-\gamma-\gamma\lambda\breve m(\lambda)|^2$, which is
\eqref{eq:d_star_frob}. The companion form follows from
$1-\gamma-\gamma\lambda\breve m(\lambda)=-\lambda\,\underline m_F(\lambda)$,
a rearrangement of \eqref{eq:companion_def} at the boundary. \qed

\subsubsection{Justification of Remark \ref{rem:equiv_inputs}}\label{subsec:proof_equiv}
The $(j,k)$ entry of $\Sn$ is
$n^{-1}\sum_i\Phi^{-1}(R_{ij}/(n+1))\Phi^{-1}(R_{ik}/(n+1))$. By the
asymptotic representation of the normal-scores rank correlation for
continuous bivariate data \citep[Thm.~2.2]{wu2022limiting} (equivalently the
nonparanormal skeptic analysis of \citealp{liu2012high}), each entry equals
$\sin(\tfrac{\pi}{2}\widehat\tau_{jk})+r_{jk}$ with
$\max_{j,k}|r_{jk}|=O_{\Prob}(n^{-1/2}\log n)$. Hence $\Sn$ and
$\widehat\bSigma^{\tau}$ agree entrywise up to $O_{\Prob}(n^{-1/2}\log n)$, and
each is an $O_{\Prob}(n^{-1/2})$-consistent entrywise estimator of $\Sigma_{jk}$
(the leading Hájek term has variance $O(n^{-1})$ and the arcsine identity of
\ref{prop:arcsine} fixes the target). This justifies the entrywise claims in Remark
\ref{rem:equiv_inputs}. 
\qed

\subsubsection{Computational complexity}\label{subsec:proof_complexity}
In the following result, we elaborate on the computational complexity of MENS. 
\begin{proposition}
	\label{prop:complexity}
\ref{alg:MENS} runs in $O\bigl(np\log n+p^2\min(n,p)\bigr)$ time and $O(p^2)$
memory.
\end{proposition}
\begin{proof}
Ranking column $j$ costs $O(n\log n)$ by merge sort; over $p$ columns this is
$O(np\log n)$, and the $np$ evaluations of $\Phi^{-1}$ cost $O(np)$. Forming
$\Sn=n^{-1}\widehat\bZ^{\top}\widehat\bZ$ costs $O(np^2)$ when $p\le n$, or
its $n\times n$ companion $n^{-1}\widehat\bZ\widehat\bZ^{\top}$ in $O(n^2 p)$
when $n<p$; the eigendecomposition costs
$O(\min(np^2,n^2p))=O(p^2\min(n,p))$ and the eigenvectors of $\Sn$ are
recovered in $O(np\min(n,p))$. The analytic shrinkage
\eqref{eq:density_est}--\eqref{eq:hilbert_est} evaluates a $p\times p$ kernel
array in $O(p^2)$; the PAV isotonization is $O(p)$. Reassembling
\eqref{eq:MENS_final} costs $O(p^3)$ if formed densely, but is never required
explicitly---downstream uses (portfolio weights, quadratic forms) act through
$\widehat\bGamma$ and $\diag(\widetilde d_j)$ in $O(p^2)$. The dominant term
is $O(np\log n+p^2\min(n,p))$, identical in order to analytic Gaussian
nonlinear shrinkage \citep{ledoit2020analytical} up to the $O(np\log n)$
ranking overhead.
\end{proof}


\section{Simulation Study}\label{sec:simulations}

We corroborate the theory by Monte Carlo. We draw a latent Gaussian vector $\by_i\sim N(\bzero,\bSigma)$ and apply several strictly increasing marginal transformations $\bm f$, so the copula is
exactly Gaussian and \ref{A1}--\ref{A5} hold. Unless stated otherwise the
latent correlation is the autoregressive $\Sigma_{jk}=0.5^{|j-k|}$. To probe
the marginal-invariance property (Theorem \ref{thm:main_RMT}) we use several marginals, including identity ($f_j=\mathrm{id}$, i.e.\ Gaussian data), mixed
(one third $t\mapsto e^{t}$, one third $t\mapsto t^{3}$, one third identity),
lognormal ($f_j=\exp$), cubic ($f_j(t)=\mathrm{sgn}(t)|t|^{3}$),
and logistic (standard-logistic margins). All produce the same latent
Gaussian correlation but wildly different observed marginals.

We compare the sample correlation (SCM); Ledoit--Wolf linear shrinkage toward
the identity (LW); Gaussian analytic nonlinear shrinkage of the sample
correlation (NLSW); Tyler's $M$-estimator of scatter followed by nonlinear
shrinkage (Tyler), a natural robust competitor for heavy-tailed data; the
sin-transformed Kendall matrix with a positive-definite projection but
no spectral correction (Kendall); and MENS. The nonlinear-shrinkage step
uses the analytic Ledoit--Wolf (2020) formula; the methods differ in the matrix
to which it is applied. 

\subsection{Marginal invariance}

According to Theorem \ref{thm:main_RMT}, the spectrum of the normal-scores covariance $\Sn$ does not depend on the marginals. We illustrate this directly. Fixing a single latent Gaussian sample and applying four different marginals, Figure \ref{fig:invariance} shows that the observed marginal distributions differ dramatically (top row) while the resulting normal-scores eigenvalue spectra are visually indistinguishable and all match the same Mar\v{c}enko--Pastur density (bottom row).

\begin{figure}[t]
\centering
\includegraphics[width=\textwidth]{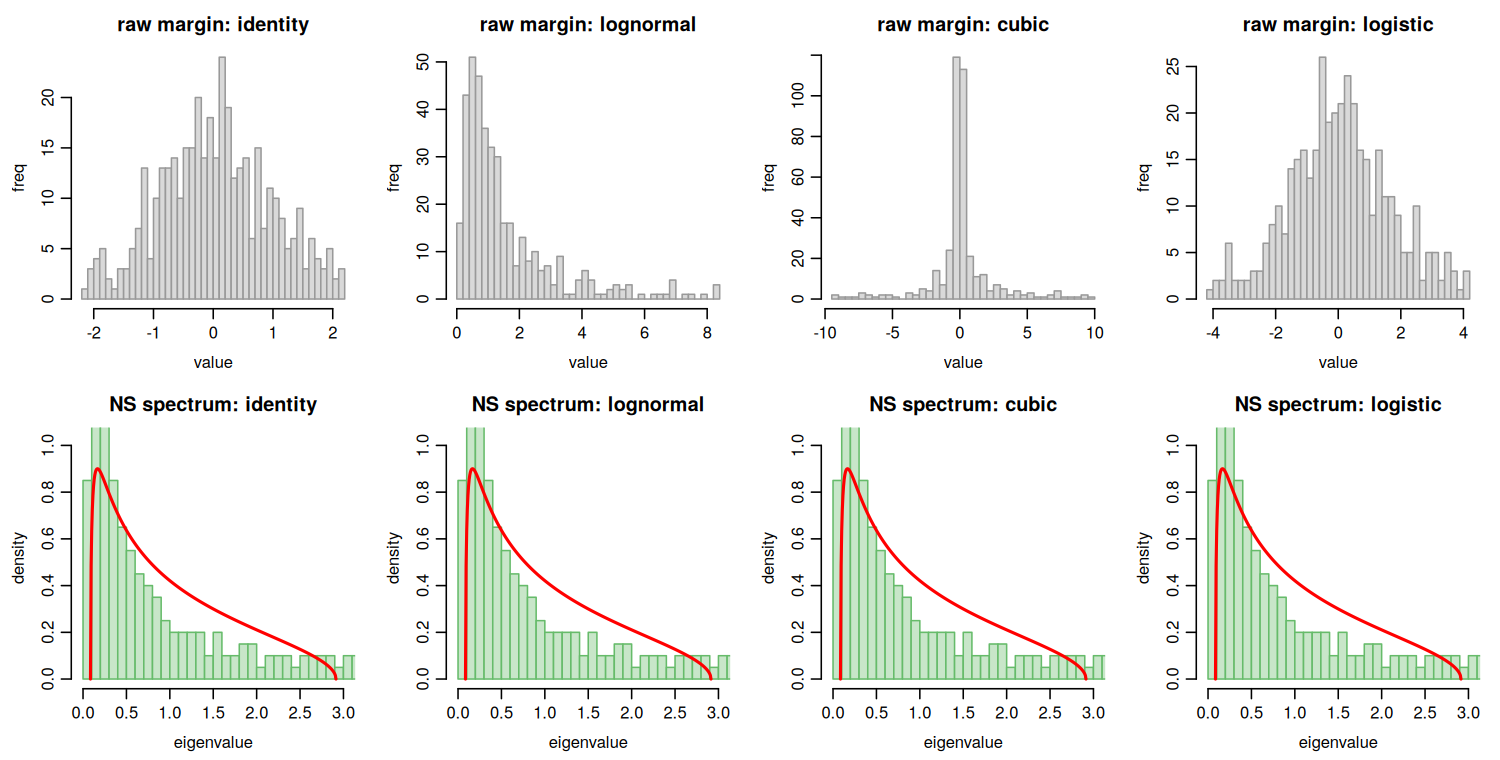}
\caption{Marginal invariance. A single latent Gaussian sample
$(p,n)=(200,400)$ is passed through four strictly increasing marginal banks.
Top: histograms of one observed coordinate---identity, lognormal, cubic, and
logistic margins differ sharply. Bottom: the eigenvalue histograms of the
corresponding normal-scores covariance $\Sn$ all coincide with the standard
Mar\v{c}enko--Pastur density (red), as Theorem \ref{thm:main_RMT} predicts.}
\label{fig:invariance}
\end{figure}

Table \ref{tab:KS} reports the Kolmogorov--Smirnov distance between
the ESD of $\Sn$ and the standard Mar\v{c}enko--Pastur law (identity latent
correlation), for three marginals and growing $(p,n)$. Within each column
the three agree to within Monte Carlo error, the invariance, and every
column decreases toward zero as $n$ grows, confirming convergence to the
predicted law.

\begin{table}[t]
\centering
\caption{Kolmogorov--Smirnov distance between the ESD of $\Sn$ and the standard
Mar\v{c}enko--Pastur law ($\bSigma=\bI_p$), for three marginal banks. Means over
$B=40$ replications. The three banks agree within each column (marginal
invariance), and all decrease with $n$.}
\label{tab:KS}
\begin{tabular}{lccccc}
\toprule
& \multicolumn{5}{c}{$(p,n)$}\\
\cmidrule(lr){2-6}
Marginal bank & $(50,200)$ & $(100,200)$ & $(150,150)$ & $(200,400)$ & $(250,500)$\\
\midrule
mixed     & 0.059 & 0.035 & 0.028 & 0.020 & 0.016\\
lognormal & 0.057 & 0.036 & 0.028 & 0.020 & 0.016\\
logistic  & 0.060 & 0.036 & 0.027 & 0.020 & 0.016\\
\bottomrule
\end{tabular}
\end{table}

Figure \ref{fig:spectral} shows the ESD against the Mar\v{c}enko--Pastur density for
two banks and two problem sizes, illustrating both the marginal invariance and
the sharpening of the fit as $n$ grows.

\begin{figure}[t]
\centering
\includegraphics[width=\textwidth]{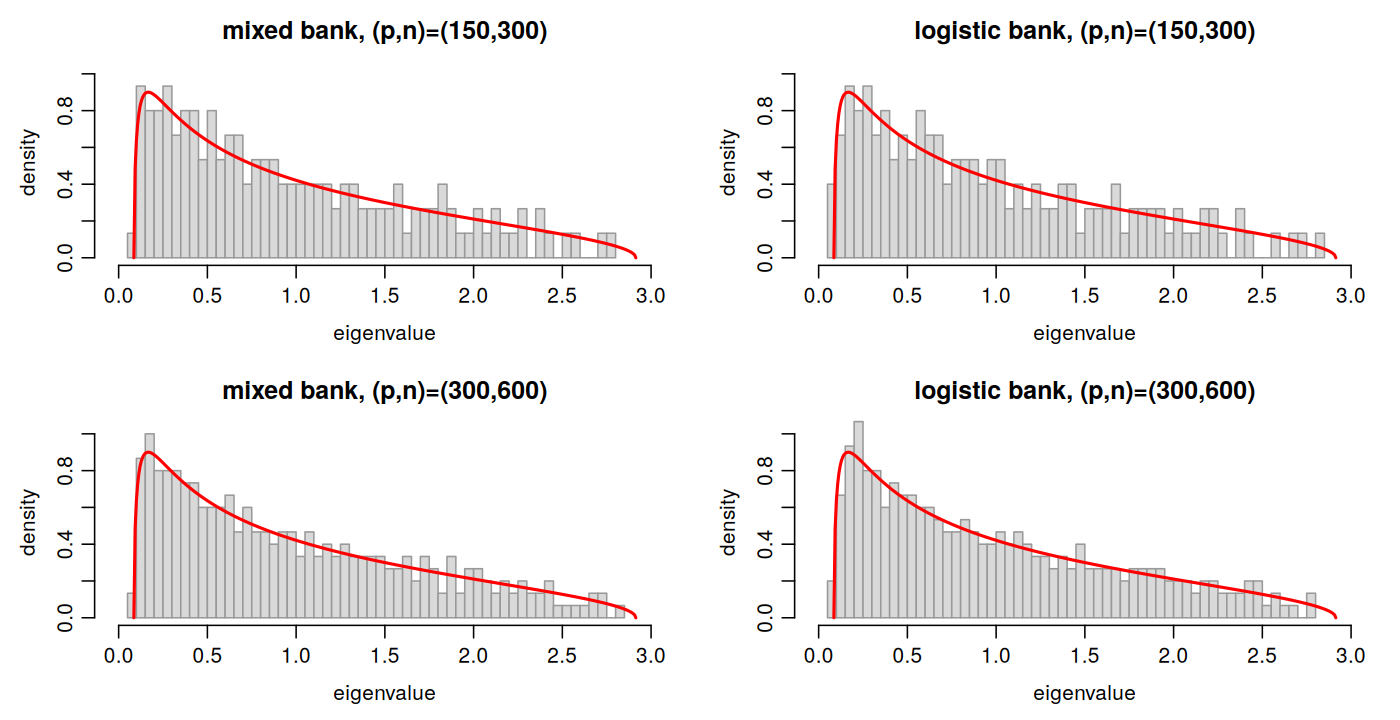}
\caption{Empirical spectral density of $\Sn$ (identity latent correlation,
$\gamma=0.5$) for the mixed and logistic marginal banks at two sizes, with the
standard Mar\v{c}enko--Pastur density overlaid (red). The fit is the same across
banks and tightens with $n$.}
\label{fig:spectral}
\end{figure} 

\subsection{Estimation accuracy}
To confirm the result of Theorem~\ref{thm:oracle_optimality}, Table \ref{tab:loss} reports operator-norm and Frobenius loss at $(p,n)=(150,300)$
($\gamma=0.5$, AR(1) latent correlation, mixed marginals), and Figure \ref{fig:accbox} shows the full replication-level distributions. MENS attains the smallest loss of all six estimators, in both norms, consistent with the oracle optimality of Theorem \ref{thm:oracle_optimality}(i); its advantage over Gaussian nonlinear shrinkage (NLSW) and the robust Tyler estimator is clear and its spread is tight. The uncorrected Kendall matrix, lacking a spectral correction, is comparatively poorly conditioned and is the worst in operator norm.

\begin{table}[t]
\centering
\caption{Operator-norm and Frobenius loss at $(p,n)=(150,300)$, AR(1) latent
correlation, mixed marginal bank. Means over $B=60$ replications (standard
deviations in parentheses for operator norm). Best in each row in \textbf{bold}.}
\label{tab:loss}
\begin{tabular}{lcccccc}
\toprule
Loss & SCM & LW & NLSW & Tyler & Kendall & \textbf{MENS}\\
\midrule
Operator & 1.946 (.099) & 1.748 (.027) & 1.709 (.034) & 1.709 (.035) & 2.330 (.122) & \textbf{1.495 (.038)}\\
Frobenius & 0.739 & 0.656 & 0.628 & 0.626 & 0.738 & \textbf{0.528}\\
\bottomrule
\end{tabular}
\end{table}

\begin{figure}[t]
\centering
\includegraphics[width=\textwidth]{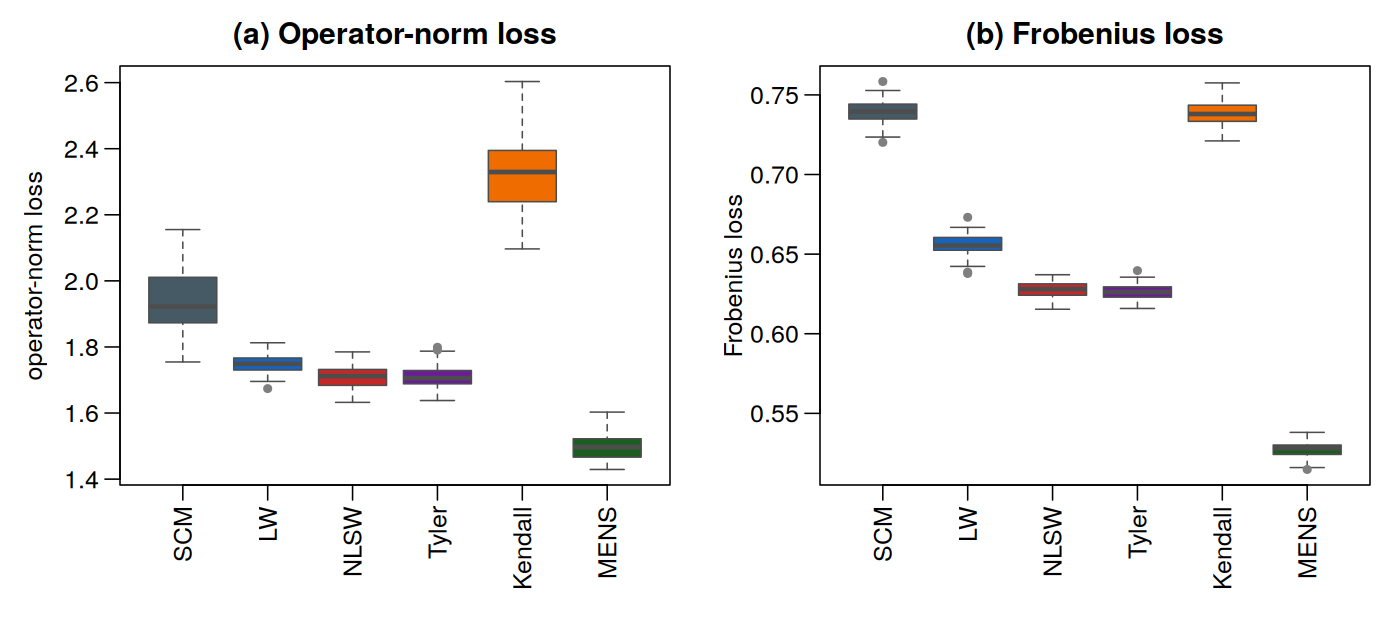}
\caption{Replication-level loss distributions at $(p,n)=(150,300)$ over $B=60$
replications: (a) operator-norm loss, (b) Frobenius loss. MENS (green) has the
lowest median and a tight spread in both panels.}
\label{fig:accbox}
\end{figure}

To illustrate the dependence on $\gamma=p/n$ we record the Frobenius loss of
MENS along six $(p,n)$ pairs with $\gamma\in[0.1,0.75]$. Regressing
$\log(\text{loss})$ on $\log\sqrt{p/n}$ gives a slope near $0.7$; the fitted
loss increases monotonically from about $0.29$ at $\gamma=0.1$ to $0.59$ at
$\gamma=0.75$. We report this as descriptive scaling, not as verification of a
rate.

\subsection{Spiked phase transition}
To check the result of Theorem~\ref{thm:BBP} numerically, we consider a single delocalized spike $\bSigma=\bI_p+\theta\bm v\bm v^{\top}$ at
$(p,n)=(200,400)$ ($\gamma=0.5$, additive threshold $\theta^{*}=\sqrt{0.5}
\approx0.707$). Figure \ref{fig:spikebox} shows, across replications,
the leading eigenvalue of $\Sn$ and its squared alignment with $\bm v$, against
the closed forms $\varphi(\theta)$ and $c(\theta)$ of \eqref{eq:identity_bulk}.
For the supercritical spikes ($\theta=2.0,1.2$) both track the theory; at and
below threshold ($\theta=\theta^{*},0.4$) the eigenvalue attaches to the bulk
edge $(1+\sqrt\gamma)^2=2.914$ and the alignment collapses toward zero, exactly
the predicted transition.

\begin{figure}[t]
\centering
\includegraphics[width=\textwidth]{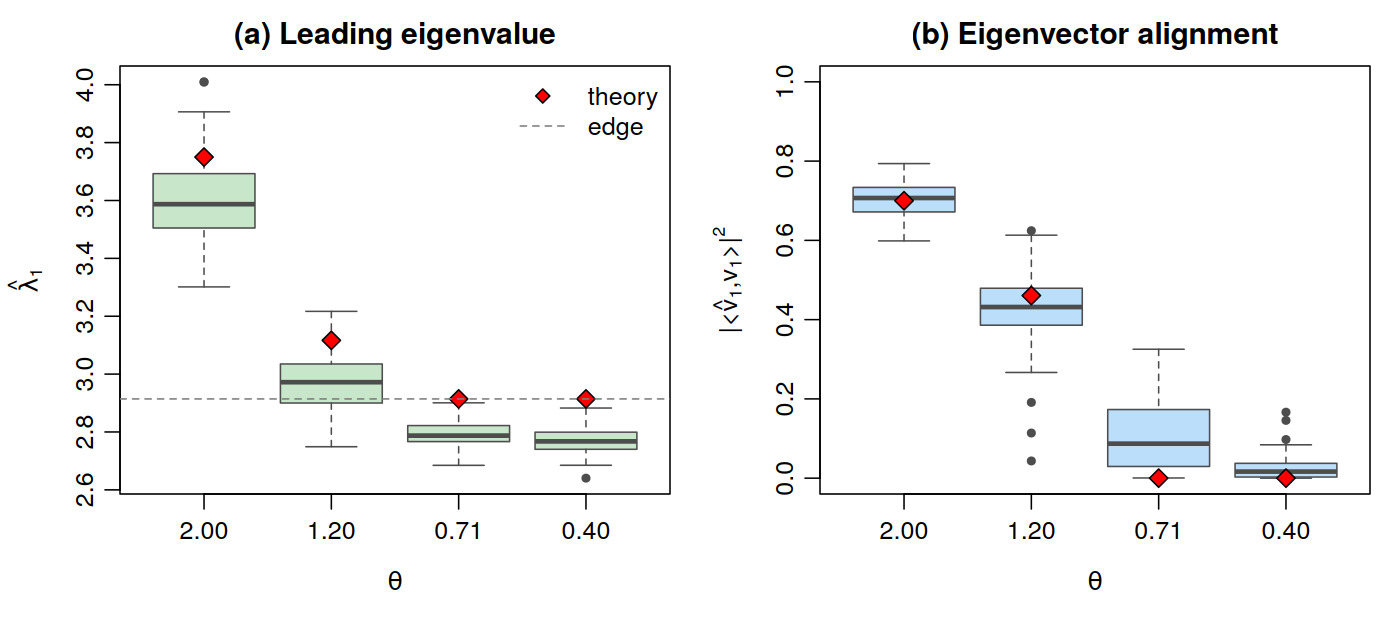}
\caption{Spiked transition on nonparanormal data, $(p,n)=(200,400)$, identity
bulk, single delocalized spike. (a) Leading eigenvalue $\hat\lambda_1$ across
$B=60$ replications at each $\theta$; red diamonds are the closed-form outlier
$\varphi(\theta)$ (or the bulk edge below threshold, dashed). (b) Squared
eigenvector alignment; red diamonds are the closed form $c(\theta)$. The
additive threshold is $\theta^{*}=\sqrt\gamma\approx0.707$.}
\label{fig:spikebox}
\end{figure}

\subsection{Minimum-variance portfolios}
Finally, we translate estimation accuracy into a downstream task. At
$(p,n)=(120,240)$ we form the global minimum-variance portfolio
$\bw=\widehat\bSigma^{-1}\bone/(\bone^{\top}\widehat\bSigma^{-1}\bone)$ from each
estimator and report the risk ratio $\mathrm{RR}=\bw^{\top}\bSigma\bw/
\bw_{\!*}^{\top}\bSigma\bw_{\!*}$ (oracle optimum $1$). Figure \ref{fig:portbox} shows
the replication-level distribution. MENS is stable and near-oracle (median
$1.11$); the uncorrected Kendall matrix, being ill-conditioned, produces
unusable portfolios (median $6.9$ with a long tail). Also, note that the linear shrinkage toward the identity (LW) achieves a slightly lower
risk ratio than MENS on this specific metric (median $1.04$), because the GMV
risk ratio rewards heavy shrinkage toward equal weights at the expense of the
estimation accuracy in which MENS is best (Table \ref{tab:loss}). The two criteria
measure different things.

\begin{figure}[t]
\centering
\includegraphics[width=0.72\textwidth]{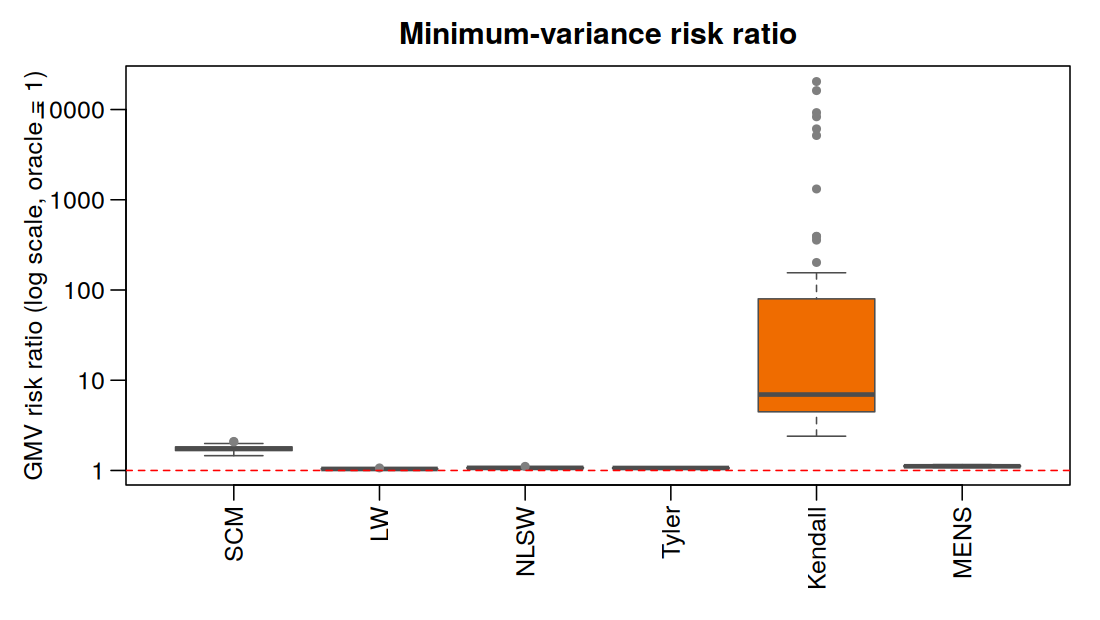}
\caption{Global minimum-variance risk ratio (log scale, oracle $=1$, red
dashed) across $B=50$ replications at $(p,n)=(120,240)$. MENS is stable and
near-oracle; the uncorrected Kendall matrix is unusable; linear shrinkage is
marginally lower than MENS on this particular metric.}
\label{fig:portbox}
\end{figure}


\section{Minimum-Variance Portfolios of S\&P\,500 Stocks}
\label{sec:realdata}

The simulation study operates within the nonparanormal model, where the theory
applies exactly. We now illustrate MENS on real equity data. We stress at the
outset that daily stock returns are heavy-tailed and exhibit tail dependence,
so they need not follow a Gaussian copula; this application therefore lies
outside the setting of Theorems \ref{thm:main_RMT}, \ref{thm:oracle_optimality}, \ref{thm:BBP} and is
intended as a practical illustration of the estimator's behaviour, not as a
validation of the theory. We plan to illustrate the two properties MENS is
designed to deliver, a well-conditioned correlation estimate and stable
eigenvalue shrinkage, translate into tangible gains for a canonical
decision-making problem in finance.

For our purpose, we use daily closing prices of the constituents of the S\&P\,500 index over
2013-02-08 to 2018-02-07, a public dataset of $505$ tickers.\footnote{The
``S\&P\,500 5-Year'' daily price panel, publicly available at
\url{https://raw.githubusercontent.com/plotly/datasets/master/all_stocks_5yr.csv}.} Retaining the $470$ names with a complete price
history and selecting the $p=200$ most liquid by median dollar volume yields a
panel of $1{,}258$ daily log-returns. This is the classical Markowitz
minimum-variance problem \citep{markowitz1952portfolio}, in which accurate,
well-conditioned covariance estimation is decisive
\citep{ledoit2004well}.

We run a rolling out-of-sample backtest. Using an estimation window of
$n=252$ trading days (about one year, so the concentration ratio is
$p/n\approx 0.79$), we form the global minimum-variance (GMV) portfolio
\begin{equation}\label{eq:gmv}
	\bw = \frac{\widehat\bSigma^{-1}\bone}{\bone^{\top}\widehat\bSigma^{-1}\bone},
\end{equation}
where $\widehat\bSigma$ is the estimated return covariance, obtained by
rescaling the correlation estimate by robust (median-absolute-deviation)
marginal volatilities as in Remark \ref{rem:correlation_covariance}. The portfolio is
held for $21$ trading days (about one month) and then rebalanced, rolling
through the sample for $47$ rebalances and $987$ out-of-sample days. We compare
MENS against Ledoit--Wolf linear shrinkage toward the identity
\citep{ledoit2004well}, the strongest and most widely used competitor for this
task. Because the GMV objective is precisely to minimize variance, the natural
success metric is the realized out-of-sample volatility of the portfolio; we
also report conditioning and turnover, which govern numerical stability and
transaction costs.

Table \ref{tab:realdata} reports the out-of-sample metrics and Figure \ref{fig:realdata}
displays the cumulative portfolio wealth and the rolling condition number of
the correlation estimate. Three findings stand out. MENS reduces the annualized out-of-sample
portfolio volatility from $11.0\%$ (linear shrinkage) to $9.4\%$, a $14.5\%$
reduction on the very objective the portfolio optimizes. The two estimators
produce nearly identical Sharpe ratios ($0.670$), so the gain is a genuine
risk reduction rather than a return artifact, and MENS buys the same
risk-adjusted return at materially lower volatility.

The median condition number of the MENS
correlation estimate is $339$, against $1{,}126$ for linear shrinkage, a
more than threefold improvement, and its median smallest eigenvalue is nearly
four times larger ($0.20$ versus $0.05$). A better-conditioned matrix inverts
more stably, which is exactly what the GMV weights \eqref{eq:gmv} require; Figure
\ref{fig:realdata}(b) shows the improvement is sustained across the whole
sample. 

Finally, MENS weights are considerably more stable across
rebalances, i.e., average turnover falls from $2.55$ to $1.50$, a $41\%$ reduction.
In practice this translates directly into lower transaction costs, an
important consideration for implementable strategies.

\begin{table}[t]
\centering
\caption{Out-of-sample performance of the global minimum-variance portfolio of
$p=200$ S\&P\,500 stocks, $2013$--$2018$, one-year estimation windows rolled
monthly ($47$ rebalances, $987$ out-of-sample days). ``Cond.'' is the median
condition number of the estimated correlation matrix and $\lambda_{\min}$ its
median smallest eigenvalue; turnover is the average one-way change in weights
per rebalance. Volatility and return are annualized.}
\label{tab:realdata}
\begin{tabular}{lcccccc}
\toprule
Estimator & Ann.\ vol. & Ann.\ ret. & Sharpe & Cond. & $\lambda_{\min}$ & Turnover\\
\midrule
LW (linear)   & 0.1105 & 0.0741 & 0.671 & 1126 & 0.054 & 2.55\\
\textbf{MENS} & \textbf{0.0945} & 0.0633 & 0.670 & \textbf{339} & \textbf{0.201} & \textbf{1.50}\\
\bottomrule
\end{tabular}
\end{table}

\begin{figure}[t]
\centering
\includegraphics[width=\textwidth]{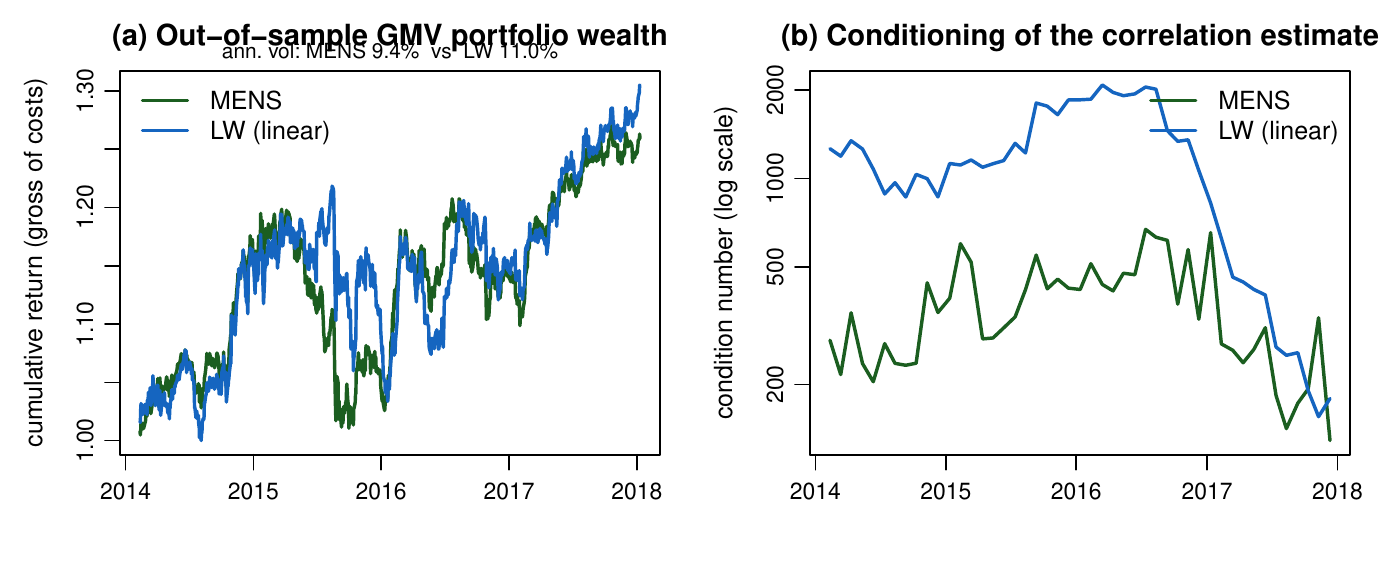}
\caption{S\&P\,500 minimum-variance backtest, $p=200$, one-year rolling
windows. (a) Cumulative out-of-sample wealth of the GMV portfolio (gross of
transaction costs); the MENS curve is visibly smoother, reflecting its lower
realized volatility ($9.4\%$ vs $11.0\%$ annualized). (b) Rolling condition
number of the estimated correlation matrix (log scale); MENS (green) is
consistently better conditioned than linear shrinkage (blue).}
\label{fig:realdata}
\end{figure}

\subsection{Interpretation and scope}

The economic reading is direct. In the high-dimensional regime $p<n$ that this
window represents, the sample correlation is invertible but severely
ill-conditioned, and its extreme eigenvalues are the ones that most distort the
GMV weights. MENS shrinks the whole eigenvalue distribution toward its
population target rather than applying a single affine map, and does so through
a rank transform that is insensitive to the skewed, heavy-tailed marginals of
returns. The result is a covariance estimate that inverts stably and yields
portfolios that are both less volatile and cheaper to trade, the two things a
minimum-variance investor cares about. MENS is most useful for
allocation problems in which the number of assets is
comparable to, but below, the length of the usable estimation window, and in
which marginal non-normality would otherwise contaminate a moment-based
estimator.

We are clear about the limits of our illustration. First, the
advantage is concentrated in the $p<n$ regime; when $p$ approaches or exceeds
$n$ the normal-scores covariance becomes singular and the analytic shrinkage
near the $p=n$ boundary is less stable, so linear shrinkage, which always
retains a strictly positive target, can be the safer choice. In unreported
experiments at $p/n\gtrsim 1$ the volatility ranking indeed reverses. Second,
returns are not a Gaussian-copula process, so the theoretical guarantees of
\ref{sec:theory} do not transfer to this application; the gains we document
are empirical. Third, the backtest is gross of transaction costs and abstracts
from short-sale and capacity constraints; the lower turnover of MENS suggests
its net-of-cost advantage would, if anything, widen, but a full cost-aware
study is beyond our scope here. Within these bounds, the application shows that
the conditioning and marginal-robustness properties MENS is built to provide
carry over from the nonparanormal model to a realistic financial task.

\section{Discussion and Conclusion}\label{sec:conclusion}

We have developed and analyzed MENS, a nonlinear shrinkage estimator of the
latent correlation matrix built on the normal-scores rank-covariance matrix,
to which it applies the Ledoit--Wolf analytic oracle shrinkage function. We demonstrated that for nonparanormal (Gaussian-copula) data the
oracle normal scores recover the latent Gaussian vector, so the
limiting spectral distribution of the rank-based matrix is the generalized
Mar\v{c}enko--Pastur law driven by the latent correlation $\bSigma$ and
invariant to the unknown marginal transformations. This is what allows
one estimator to be simultaneously robust to arbitrary monotone marginals (the
hallmark of rank methods) and asymptotically efficient (the hallmark of
nonlinear shrinkage) within this class. We established Frobenius oracle
optimality among rotation-equivariant estimators, operator-norm consistency,
and a parametric $n^{-1/2}$ minimax lower bound (Theorem \ref{thm:oracle_optimality}),
together with a closed-form spiked phase transition (Theorem \ref{thm:BBP}). The exact
recovery of the latent Gaussian vector is specific to the Gaussian copula; a
non-Gaussian dependence structure would change the population object driving
the spectrum.

During our study, we faced with some technical improvement but did not elaborated here. In this resepct, we raise three questions, left open as conjectures.

\begin{conjecture}[Quantitative spectral rate]\label{conj:rate}
Under \ref{A1}--\ref{A4}, there is a constant $C>0$ such that, for every fixed
$\eta_0>0$,
$\sup_{\Im z\ge\eta_0}|m_n(z)-m_F(z)|=O_{\Prob}(n^{-1/2})$, and more generally
a local Mar\v{c}enko--Pastur law holds for the normal-scores covariance $\Sn$
down to spectral scales $\Im z\gtrsim n^{-1+\varepsilon}$. Establishing this
would likely require an anisotropic/local-law analysis of $\Sn$ that tracks
the rank-perturbation error of \ref{lem:rank_perturbation} inside the
resolvent.
\end{conjecture}

\begin{conjecture}[Sharp operator-norm minimax rate]\label{conj:minimax}
Over the bounded-spectrum correlation class
$\mathcal C(\sigma_{\min},\sigma_{\max})$, the nonparanormal operator-norm
minimax risk satisfies
$\inf_{\widehat\bSigma}\sup_{\bSigma\in\mathcal C}
\E\normop{\widehat\bSigma-\bSigma}\asymp\sqrt{p/n}$, with the lower bound
attained by a delocalized-block Assouad construction and the upper bound
attained by MENS via an operator-norm control of the oracle-shrinkage bias.
Theorem \ref{thm:oracle_optimality}(iii) proves only the weaker $n^{-1/2}$ floor.
\end{conjecture}

\begin{conjecture}[Spectral equivalence of rank inputs]\label{conj:equiv}
The normal-scores covariance $\Sn$ and the sin-transformed Kendall matrix
$\widehat\bSigma^{\tau}$ have the same limiting empirical spectral distribution
under $p/n\to\gamma$, despite $\normop{\Sn-\widehat\bSigma^{\tau}}$ not
vanishing. A proof would need to go beyond the operator-norm and normalized
Frobenius controls used in Remark \ref{rem:equiv_inputs}, which our simulations show
are insufficient.
\end{conjecture}

Beyond these, several extensions are natural. A factor-augmented variant that
treats a fixed number of spiked latent directions by the explicit forward map
of Theorem \ref{thm:BBP} and shrinks only the bulk would sharpen estimation when a
strong factor is present. A central limit theorem for linear spectral
statistics of $\Sn$ would yield formal, marginal-invariant tests of
correlation structure. Finally, extending the estimator beyond the Gaussian
copula, for a specified non-Gaussian dependence, by inverting the induced
score-correlation map, or by anchoring directly to $\sin(\tfrac\pi2\widehat
\tau)$ via the arcsine identity of \ref{prop:arcsine}, is the natural route
to guarantees outside the nonparanormal class, and we regard it as the most
important problem left open by this work. Thus, extending this work to meta-elliptical family is worthwhile. 

\section*{Acknowledgements}
This research was supported in part by the Iran National Science Foundation
(INSF) under Grant No.\ 4015320.

\section*{Data and Code Availability}
All numbers, tables, and figures are produced by a self-contained
reference implementation in \texttt{R}, provided at \url{https://github.com/M-Arashi/Shrinkage.git}, subdirectory ``MENS''. The simulation study uses only synthetic data generated by the accompanying scripts. The real-data application in \ref{sec:realdata} uses the public ``S\&P\,500 5-Year'' daily price panel (\url{https://raw.githubusercontent.com/plotly/datasets/master/all_stocks_5yr.csv}).




\end{document}